\documentclass[a4paper,10pt,authoryear]{elsarticle}
\usepackage[colorlinks]{hyperref}
\usepackage{microtype}
\usepackage{enumitem}
\usepackage[x11names]{xcolor}
\usepackage{amsthm}
\usepackage{graphicx}
\usepackage{amsmath}
\usepackage{amssymb}
\usepackage{booktabs}
\usepackage{caption}
\usepackage{subcaption}

\journal{Journal of Mathematical Economics}

\newcommand{\R}{\mathbb{R}}
\newcommand{\N}{\mathbb{N}}

\newcommand{\E}{\mathbb{E}}

\definecolor{cite}{HTML}{b2103d}
\definecolor{link}{HTML}{0a3ece}

\newtheorem{thm}{Theorem}[section]
\newtheorem{lemma}[thm]{Lemma}

\newtheorem{prop}[thm]{Proposition}
\theoremstyle{definition}
\newtheorem{assumption}[thm]{Assumption}

\begin{document}
\begin{frontmatter}
  \title{ Equilibrium in Production Chains with Multiple Upstream
    Partners\tnoteref{titlenote}}

\tnotetext[titlenote]{The authors thank John Stachurski for his helpful
  comments and valuable suggestions. We also thank Simon Grant, Ronald
  Stauber, Ruitian Lang, and seminar participants at the Australian National
  University for their helpful comments and discussions. We are also
  grateful for the inputs from the editor and two anonymous referees. This
  research is supported by an Australian Government Research Training
  Program (RTP) Scholarship.}
\date{February 16, 2019}
\author[ad1,ad2]{Meng Yu}
\ead{yumeng161@mails.ucas.ac.cn}
\address[ad1]{Institute of Mathematics, Academy of Mathematics and Systems
  Science, Chinese Academy of Sciences, Beijing 100190, China}
\address[ad2]{School of Mathematical Sciences, University of Chinese Academy
  of Sciences, Beijing 100049, China}

\author[add]{Junnan Zhang \corref{mycorrespondingauthor}}
\cortext[mycorrespondingauthor]{Corresponding author}
\ead{junnan.zhang@anu.edu.au}
\address[add]{Research School of Economics, Australian National University,
  Australia}

\begin{abstract}
  In this paper, we extend and improve the production chain model introduced
  by \cite{kikuchi2018span}. Utilizing the theory of monotone concave
  operators, we prove the existence, uniqueness, and global stability of
  equilibrium price, hence improving their results on production networks
  with multiple upstream partners. We propose an algorithm for computing the
  equilibrium price function that is more than ten times faster than
  successive evaluations of the operator. The model is then generalized to a
  stochastic setting that offers richer implications for the distribution of
  firms in a production network.
\end{abstract}


\end{frontmatter}

\section{Introduction}

Over the past several centuries, firms have self-organized into ever more
complex production networks, spanning both state and international
boundaries, and constructing and delivering a vast range of manufactured
goods and services. The structures of these networks help determine the
efficiency \citep{levine2012production, ciccone2002input} and resilience
\citep{carvalho2007aggregate, jones2011intermediate, bigio2016financial,
  acemoglu2012network, acemoglu2015networks} of the entire economy, and also
provide new insights into the directions of trade and financial policies
\citep{baldwin2013spiders, acemoglu2015systemic}.

We consider a production chain model introduced by
\cite{kikuchi2018span} that examines the formation of such structures.
They connect the literature on firm networks and network structure to the
underlying theory of the firm. In particular, the model in
\cite{kikuchi2018span} formalizes the foundational ideas on the theory
of the firm presented in \cite{coase1937nature}, embedding them in an
equilibrium model with a continuum of price taking firms, and providing
mathematical representations of the determinants of firm boundaries
suggested by \cite{coase1937nature}.

A single firm at the end of the production chain sells a final product to
consumers. The firm can choose to produce the whole product by itself or
subcontract a portion of it to possible multiple upstream partners, who then
make similar choices until all the remaining production is completed. The
main reason for firms to produce more in-house is to save the transaction
costs of buying intermediate products from the market. In fact,
\cite{coase1937nature} regards this as the primary force that brings
firms into existence. An opposing force that limits the size of a firm is
the costs of organizing production within the firm\footnote{One
  justification also mentioned in \cite{kikuchi2018span} is that firms
  usually experience diminishing return to management: when a firm gets
  bigger it also bears increasing coordination costs. See also
  \cite{coase1937nature}, \cite{lucas1978size}, and
  \cite{becker1992division}.}. A price function governs the choices
firms make and is determined endogenously in equilibrium when every firm in
the production chain makes zero profit.

Considering that all firms are ex ante identical, a notable feature of this
model is its ability to generate a production network with multiple layers
of firms different in their sizes and numbers of upstream partners. The
source of the heterogeneity lies solely in the transaction costs and firms'
different stages in the production chain. This feature provides insights
into the formation of potentially more complex structures in a production
network. \cite{kikuchi2018span} prove the existence, uniqueness, and
global stability\footnote{Mathematically, the equilibrium price function is
  determined as the fixed point of a Bellman like operator (see
  Section~\ref{sec:eq}). Globally stability means that the fixed point can
  be computed by successive evaluations of the operator on any function in a
  certain function space.} of the equilibrium price function restricting
every firm to have only one upstream partner. In this case, the resulting
production network consists of a single chain.

There are however, several significant weaknesses with the analysis in
\cite{kikuchi2018span}. First, while they provide comprehensive results on
uniqueness of equilibrium prices and convergence of successive
approximations in the single upstream partner case, they fail to provide
analogous results for the more interesting multiple upstream partner case,
presumably due to technical difficulties. Second, their model cannot
accurately reflect the data on observed production networks because their
networks are always symmetric, with sub-networks at each layer being exact
copies of one another. Real production networks do not exhibit this
symmetry\footnote{For instance, for a mobile phone manufacturer, most
  subcontractors who supply complicated components like display or CPUs have
  multiple upstream partners of their own, while those who supply raw
  materials usually don't \citep{dedrick2011distribution,
    kraemer2011capturing}.}. Third, they provide no effective algorithm for
computing the equilibrium price function in the multiple upstream partner
case.

This paper resolves all of the shortcomings listed above. As our first
contribution, we extend their existence, uniqueness, and global stability
results to the multiple partner case. To avoid the technical difficulties
faced in their paper, we employ a different approach utilizing the theory of
monotone concave operators, which enables us to give a unified proof for
both cases.

Theoretically, the concave operator theory ensures the global stability of
the fixed point, so the equilibrium price function can be computed by
successive evaluations of the operator. In practice, however, the rates of
convergence can be different for different model settings. This leads to
unnecessarily long computation time in most cases. As a second contribution,
we propose an algorithm that achieves fast computation regardless of
parameterizations and is shown to drastically reduce computation time in our
simulations.

A third contribution of this paper is that we generalize the model to a
stochastic setting. In the original model, the equilibrium firm allocation
is symmetric and deterministic: firms at the same stage of production choose
the exact same number of upstream partners. In reality, each firm faces
uncertainty in the contracting process and cannot always choose the optimal
number of partners. We model the number of upstream partners as a Poisson
distribution and let the firm choose its parameter, which can be seen as a
search effort. Using the same approach, we prove the existence and
uniqueness of equilibrium price function as well as the validity of the
algorithm. We further use simulations to analyze how production and
transaction costs determine the shape of a production network. This
generalization provides a new source of heterogeneity in the equilibrium
firm allocation and can be a potential channel for future research on size
distribution of firms.

As briefly mentioned above, the method we use to establish the existence,
uniqueness, and global stability of the equilibrium price function draws on
the theory of concave operators. A competing method traditionally used for
the same purpose is the Contraction Mapping Theorem, which has been an
essential tool for economists dealing with various dynamic models ever since
\cite{bellman1957dynamic}. So long as the operator in question satisfies
the contraction property, we can quickly compute a unique fixed point by
applying the operator successively. This property, simple as it may be, is
not shared among a number of important models, urging us to find new tools
to tackle fixed point problems in economic dynamics.

The theory of monotone concave operators originally due to \citet[Chapter
6]{krasnoselskii1964positive} is another simple yet powerful tool. The idea
behind it is intuitive: imagine an increasing and strictly concave real
function $f$ such that $f(x_1) > x_1$ and $f(x_2) < x_2$ with $x_1 < x_2$.
Then it must be true that f has a unique fixed point on $[x_1, x_2]$, and by
the concavity of $f$, the fixed point can be computed by successive
evaluations of $f$ on any $x\in[x_1, x_2]$. No contraction property is
needed here while we still get all the results from the Contraction Mapping
Theorem. A full-fledged theorem owing to \cite{du1989fixed} for arbitrary
Banach spaces is stated in Theorem~\ref{thm:du}. For similar
treatments\footnote{We thank an anonymous referee for referring us to some
  of the works mentioned here.} in the math literature, also see
\cite{krasnoselskii1972approximate}, \cite{krasnosel1984geometrical},
\cite{guo1988nonlinear}, \cite{guo2004partial}, and
\cite{zhang2012variational}.

Apart from Theorem~\ref{thm:du}, there are other similar techniques that
utilize concavity to show uniqueness\footnote{The existence of fixed points
  can be tackled in various ways. For the operator $T$, existence has
  already been proved in \cite{kikuchi2018span} who use the classical
  Knaester--Tarski fixed point theorem. Alternatively, it can be proved, as
  an anonymous referee suggests, by demonstrating that $T$ is completely
  continuous \citep[see, e.g.,][Chapter 4]{krasnoselskii1964positive}. The
  Schauder typed fixed point theorems also apply; see, for example, Section
  7.1 in \cite{cheney2013analysis}.} of the fixed point.
\cite{krasnoselskii1964positive} shows that a monotone operator on a
positive cone has at most one nonzero fixed point if the operator satisfies
a concavity condition ($u_0$-concave).
For applications of this technique in the economic literature, see, for
example, \cite{lacker1991money} and \cite{becker2017recursive}. Following
\cite{krasnoselskii1964positive}, \cite{coleman1991equilibrium} proves
uniqueness under slightly different concavity and monotonicity conditions
(pseudo-concave and $x_0$-monotone).
See also \cite{datta2002existence}, \cite{datta2002monotone}, and
\cite{morand2003existence} for other economic applications along this line.

\cite{marinacci2010unique, marinacci2017unique} link concavity to
contraction in the Thompson metric \citep{thompson1963certain}, which allows
one to apply the Contraction Mapping Theorem to operators that are not
contractive under the supnorm.
In a similar vein, \cite{marinacci2017unique} establish existence and
uniqueness results for monotone operators under a range of weaker concavity
conditions using Tarski-type fixed point theorems and the Thompson metric.
Among all of these results, the theorem by \cite{du1989fixed} turns out to
be the most suitable for our work.




The monotone concave operator theory has seen some recent success in the
economic literature. \cite{lacker1991money} study an economy with cash and
trade credit as means of payment and show that the equilibrium interest rate
is a unique fixed point of a monotone concave operator.
\cite{coleman1991equilibrium, coleman2000uniqueness} studies the equilibrium
in a production economy with income tax and proves the existence and
uniqueness of consumption function by constructing a monotone concave map.
Following this approach, \cite{datta2002existence} prove the existence and
uniqueness of equilibrium in a large class of dynamic economies with capital
and elastic labor supply. Similar work in the same vein includes
\cite{morand2003existence} and \cite{datta2002monotone}.
\cite{rincon2003existence} exploit the monotonicity and convexity properties
of the Bellman operator and give conditions for existence and uniqueness of
fixed points in the case of unbounded returns. \cite{balbus2013constructive}
study the existence and uniqueness of pure strategy Markovian equilibrium
using theories concerning decreasing and ``mixed concave'' operators. More
recently, this theory has been applied extensively to models with recursive
utilities since \cite{marinacci2010unique}; other
contributions\footnote{Among these works, \cite{balbus2016non},
  \cite{borovivcka2017necessary}, and \cite{ren2018dynamic} use versions of
  fixed point theorems similar to Theorem~\ref{thm:du} in this paper.} along
this line include \cite{balbus2016non}, \cite{borovivcka2017necessary,
  borovivcka2018existence}, \cite{becker2017recursive},
\cite{marinacci2017unique}, \cite{pavoni2018dual}, \cite{bloise2018convex},
and \cite{ren2018dynamic}.

Our work connects to this literature in that the operator which determines
the equilibrium price is shown to be increasing and concave but does not
satisfy any contraction property. To prove existence and uniqueness,
\cite{kikuchi2018span} use an ad hoc and convoluted method for the case
when every firm can only have one upstream partner but fail to generalize it
to the multiple partner case. Using the monotone concave operator theory, we
are able to extend their results and give a much simpler proof.

Section~\ref{sec:model} describes the model in detail. Section~\ref{sec:eq}
introduces the monotone concave operator theory and gives existence and
uniqueness results. The algorithm is described in Section~\ref{sec:com}.
Section~\ref{sec:st} generalizes the model, allowing for stochastic choices
of upstream partners. Section~\ref{sec:con} concludes. All proofs can be
found in the \hyperref[app:1]{Appendix}.

\section{The Model}
\label{sec:model}

We study the production chain model with multiple partners in
\cite{kikuchi2018span}. The chain consists of a single firm at the end of
the chain which sells a single final good to consumers and firms at
different stages of the production, each of which sells an intermediate good
to a downstream firm by producing the good in-house or subcontracting a
portion of the production process to possibly multiple upstream firms. We
index the stage of production by $s\in X = [0, 1]$ with 1 being the final
stage. Each firm faces a price function $p: X \to \R_+$ and a cost function
$c: X \to \R_+$. Subcontracting incurs a transaction cost that is
proportionate\footnote{Here we follow \cite{kikuchi2018span}. This
  transaction cost can be the cost of gathering information, drafting
  contract, bargaining, or even tax, all of which tend to increase with the
  volume of the transaction.} to the price with coefficient $\delta > 1$ for
each upstream partner and an additive transaction cost $g: \N \to \R_+$ that
is a function of the number of upstream partners. The cost $g$ can be seen
as the costs of maintaining partnerships such as legal expenses and
communication costs.

We adopt the same assumptions as in \cite{kikuchi2018span}. For the cost
function $c$, we assume that $c(0) = 0$ and it is differentiable, strictly
increasing, and strictly convex. In other words, each firm experiences
diminishing return to management as mentioned in the introduction. This
assumption is needed here because otherwise no firm would want to
subcontract its production. We also assume $c'(0) > 0$. For the additive
transaction cost function $g$, we assume that it is strictly increasing,
$g(1) = 0$, and $g(k)$ goes to infinity as the number of upstream partners
$k$ goes to infinity. To summarize, we have the following two assumptions.

\begin{assumption}
  \label{c}
  The cost function $c$ is differentiable, strictly increasing, and strictly
  convex. It also satisfies $c(0) = 0$ and $c'(0) > 0$.
\end{assumption}

\begin{assumption}
  \label{g}
  The additive transaction cost function $g$ is strictly increasing, $g(1) =
  0$, and $g(k) \to \infty$ as $k\to \infty$.
\end{assumption}

Therefore, a firm at stage $s$ solves the following problem:
\begin{equation}
  \label{eq:prob}
  \min_{\substack{t\leq s\\ k\in\N}}\left\{ c(s - t) + g(k) + \delta k p(t/k)
  \right\}.
\end{equation}
In (\ref{eq:prob}), the firm chooses to produce $s-t$ in-house with cost
$c(s-t)$ and subcontract $t$ to $k$ upstream partners. Since each
subcontractor is in charge of $t/k$ part of the product, this results in a
proportionate transaction cost $\delta k p(t/k)$ and an additive transaction
cost $g(k)$. Then the firm sells the product to its downstream firm at price
$p(s)$.

\section{Equilibrium}
\label{sec:eq}
Following \cite{kikuchi2018span}, we consider the equilibrium in a
competitive market with free entry and free exit. The price adjusts so that
in the long run every firm makes zero profit. The equilibrium price function
then satisfies
\begin{equation}
  p(s) = \min_{\substack{t\leq s\\ k\in\N}}\left\{ c(s - t) + g(k) + \delta k
    p(t/k) \right\}. 
\end{equation}
Let $R(X)$ be the space of real functions and $C(X)$ the space of continuous
functions on $X$. Then we can define an operator $T: C(X) \to R(X)$ by
\begin{equation}
  \label{eq:T}
  Tp(s) := \min_{\substack{t\leq s\\ k\in\N}}\left\{ c(s - t) + g(k) + \delta k
    p(t/k) \right\}. 
\end{equation}
The equilibrium price function is thus determined as the fixed point of the
operator $T$.

\subsection{Monotone Concave Operator Theory}

Before proceeding to our main result, we first introduce a theorem due to
\cite{du1989fixed}, which studies the fixed point properties of monotone
concave operators on a partially ordered Banach space.

Let $E$ be a real Banach space on which a partial ordering is defined by a
cone $P \subset E$, in the sense that $x \leq y$ if and only if $y - x\in
P$. If $x \leq y$ but $x \neq y$, we write $x < y$. An operator $A: E\to E$
is called an \emph{increasing} operator if for all $x, y\in E$, $x \leq y$
implies that $Ax \leq Ay$. It is called a \emph{concave} operator if for any
$x, y\in E$ with $x \leq y$ and any $t\in [0, 1]$, we have $A\left(tx +
  (1-t)y\right) \geq tAx + (1-t)Ay$. For any $u_0, v_0 \in E$ with $u_0 <
v_0$, we can define an order interval by $[u_0, v_0] := \left\{ x\in E: u_0
  \leq x \leq v_0 \right\}$. We have the following theorem (\citealp[see,
e.g.,][Theorem 3.1.6]{guo2004partial} or \citealp[Theorem
2.1.2]{zhang2012variational}).

\begin{thm}[\citealp{du1989fixed}]
  \label{thm:du}
  Suppose $P$ is a normal cone\footnote{A cone $P\subset E$ is said to be
    normal if there exists $\delta>0$ such that $\|x + y\| \geq \delta$ for
    all $x, y\in P$ and $\|x\| = \|y\| = 1$.}, $u_0, v_0 \in E$, and $u_0 <
  v_0$. Moreover, $A: [u_0, v_0] \to E$ is an increasing operator. Let $h_0 =
  v_0 - u_0$. If $A$ is an concave operator, $Au_0 \geq u_0 + \epsilon h_0$
  for some $\epsilon \in [0, 1]$, and $Av_0 \leq v_0$, then $A$ has a unique
  fixed point $x^*$ in $[u_0, v_0]$. Furthermore, for any $x_0\in [u_0, v_0]$,
  $A^n x_0 \to x^*$ as $n\to \infty$.
\end{thm}

This theorem gives a sufficient condition for the existence, uniqueness, and
global stability of the fixed point of an operator without assuming it to be
a contraction mapping. It is particularly useful in cases where we study a
monotone concave operator but the contraction property is hard or impossible
to establish. This is the case in our model. The operator $T$ is not a
contraction\footnote{To be more rigorous, $T$ is not a contraction under the
  supnorm, but it might be a contraction in some other complete metric. In
  fact, \cite{bessaga1959converse} proves a partial converse of the
  Contraction Mapping Theorem, which ensures that under certain conditions
  there exists a complete metric in which $T$ is a contraction. Also see
  \cite{leader1982uniformly}; for the construction of such metrics, see
  \cite{janos1967converse} and \cite{williamson1987constructing}. For an
  application of this theorem in the economic literature, see
  \cite{balbus2013constructive}. We wish to thank an anonymous referee for
  referring us to this literature.} because the transaction cost coefficient
$\delta$ is greater than 1, but as will be shown below, $T$ is actually an
increasing concave operator.

Based on Theorem~\ref{thm:du}, we have the following theorem.

\begin{thm}
  \label{main}
  Let $u_0(s) = c'(0)s$, $v_0(s) = c(s)$, and $[u_0, v_0]$ be the order
  interval on $C(X)$ with the usual partial order. If Assumption \ref{c} and
  \ref{g} hold, then $T$ has a unique fixed point $p^*$ in $[u_0,
  v_0]$. Furthermore, $T^n p \to p^*$ for any $p\in [u_0, v_0]$.
\end{thm}

This theorem ensures that there exists a unique price function in
equilibrium and it can be computed by successive evaluation of the operator
$T$ on any function located in that order interval\footnote{For the choice
  of the order interval we also follow \cite{kikuchi2018span}.}.
Furthermore, as is clear in the proof (see \ref{app:1}), the existence of
the minimizers $t^*(s)$ and $k^*(s)$ can also be proved, although they might
not be single valued for some $s$.

\subsection{Properties of the Solution}

In the case where each firm can only have one upstream partner, the
equilibrium price function is strictly increasing and strictly convex
\citep{kikuchi2018span}. In this model, however, complications arise since
firms at different stages might choose to have different numbers of upstream
partners. In fact, the equilibrium price is usually piece-wise convex due to
this fact. An example\footnote{The parameterization here is merely chosen to
  highlight the shape of the price function and is not economically
  realistic. The price is computed using a faster algorithm introduced in
  Section \ref{sec:com} with $m = 5000$ grid points instead of successive
  evaluation of $T$.}  of the equilibrium price function is plotted in Figure
\ref{fig:price_ex} where $c(s) = e^{10s} - 1$, $g(k) = \beta(k - 1)$ with
$\beta=50$, and $\delta = 10$. As is shown in the plot, the price function as
a whole is not convex, but it is piece-wise convex with each piece
corresponding to a choice of $k$. Monotonicity of $p^*$ remains true.

\begin{prop}
  \label{p_inc}
  The equilibrium price function $p^*:X\to \R_+$ is strictly increasing.
\end{prop}

As for comparative statics, we have some basic results also present in
\cite{kikuchi2018span} about the effect of changing transaction costs on the
equilibrium price function. If either transaction cost ($\delta$ or $g$)
increases, the equilibrium price function also increases.

\begin{prop}
  \label{comp}
  If $\delta_a \leq \delta_b$, then $p^*_a \leq p^*_b$. Similarly, if $g_a
  \leq g_b$, $p^*_a \leq p^*_b$.
\end{prop}

In Figure~\ref{fig:price_comp}, we plot how the equilibrium price function
changes when transaction cost increases. The baseline model setting is the
same as Figure~\ref{fig:price_ex}. We can see that if $\delta$ or $\beta$
increases, the equilibrium price function also increases.

\begin{figure}[tb!]
  \centering
  \includegraphics[width=0.7\textwidth]{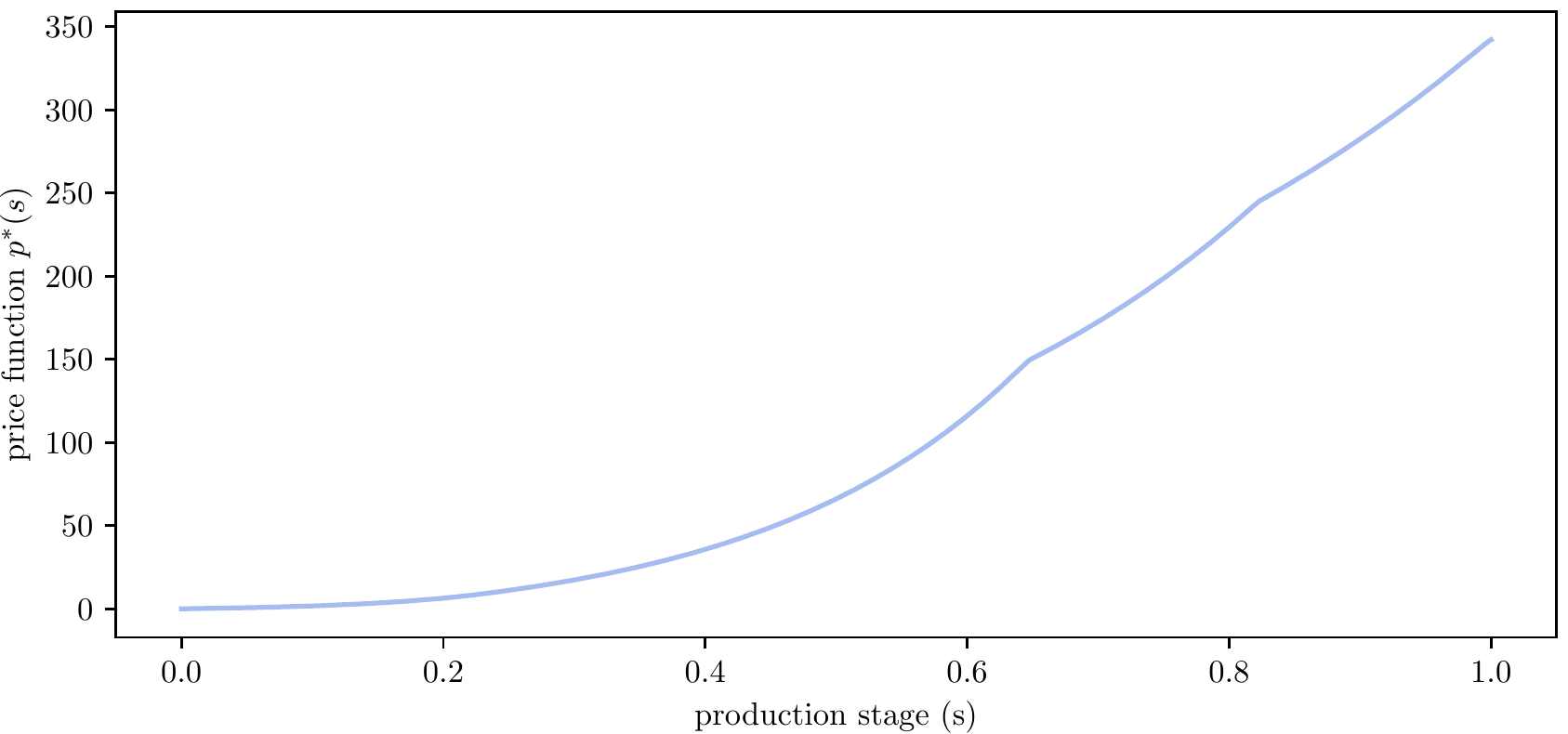}
  \caption{An example of equilibrium price function.}
  \label{fig:price_ex}
\end{figure}

\begin{figure}[tb!]
  \centering
  \includegraphics[width=0.7\textwidth]{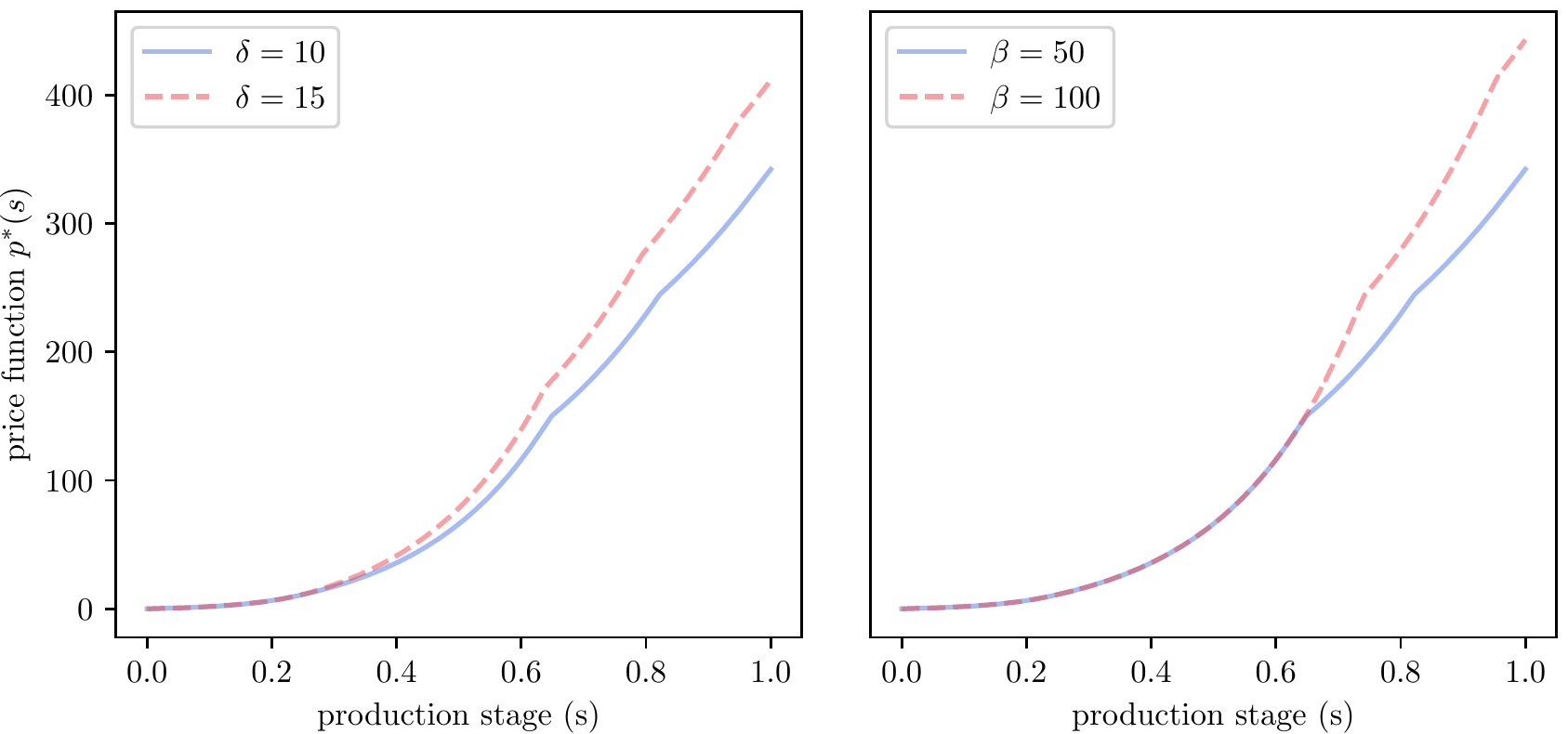}
  \caption{Equilibrium price function when $c(s) = e^{10s} - 1$ and $g(k) =
    \beta(k-1)$.}
  \label{fig:price_comp}
\end{figure}

\section{Computation}
\label{sec:com}

To compute an approximation to the equilibrium price function given
$\delta$, $c$, and $g$, one possibility is to take a function in $[u_0,
v_0]$ and iterate with $T$. However, in practice we can only approximate the
iterates, and, since $T$ is not a contraction mapping the rate of
convergence can be unsatisfactory for some model settings. On the other
hand, as we now show, there is a fast, non-iterative alternative that is
guaranteed to converge.

\begin{table}[b!]
  \centering
  \begin{tabular*}{\textwidth}{@{}l}
    \toprule
    \textbf{Algorithm 1} Construction of $p$ from $G = \{0, h, 2h, . . . ,
    1\}$\\ 
    \midrule
    \quad $p(0) \gets 0$\\
    \quad $s \gets h$\\
    \quad \textbf{while} $s\leq 1$ \textbf{do}\\
    \qquad evaluate $p(s)$ via equation (\ref{eq:alg})\\
    \qquad define $p$ on $[0, s]$ by linear interpolation of $p(0), p(h),
    p(2h), \ldots, p(s)$\\
    \qquad $s \gets s + h$\\
    \quad \textbf{end while}\\
    \bottomrule
  \end{tabular*}
\end{table}

Let $G = \{0, h, 2h, . . . , 1\}$ for fixed $h$. Given $G$, we define our
approximation $p$ to $p^*$ via the recursive procedure in Algorithm 1. In the
fourth line, the evaluation of $p(s)$ is by setting
\begin{equation}
  \label{eq:alg}
  p(s) = \min_{\substack{t\leq s - h\\ k\in\N}}\left\{ c(s - t) + g(k) +
    \delta k p(t/k) \right\}. 
\end{equation}
In line five, the linear interpolation is piecewise linear interpolation of
grid points $0, h, 2h, \ldots, s$ and values $p(0), p(h), p(2h), \ldots,
p(s)$.

The procedure can be implemented because the minimization step on the
right-hand side of (\ref{eq:alg}), which is used to compute $p(s)$, only
evaluates $p$ on $[0, s - h]$, and the values of $p$ on this set are
determined by previous iterations of the loop. Once the value $p(s)$ has been
computed, the following line extends $p$ from $[0, s - h]$ to the new interval
$[0, s]$. The process repeats. Once the algorithm completes, the resulting
function $p$ is defined on all of $[0, 1]$ and satisfies $p(0) = 0$ and
(\ref{eq:alg}) for all $s \in G$ with $s > 0$.

Now consider a sequence of grids $\{G_n\}$, and the corresponding functions
$\{p_n\}$ defined by Algorithm 1. Let $G_n = \{0, h_n, 2h_n, \ldots , 1\}$
with $h_n = 2^{-n}$. In this setting we have the following result, the proof
of which is given in \ref{app:2}.

\begin{thm}
  \label{alg_conv}
  If Assumption \ref{c} and \ref{g} hold, then $\{p_n\}$ converges to $p^*$
  uniformly.
\end{thm}

\begin{figure}[t]
  \centering
  \includegraphics[width=0.7\textwidth]{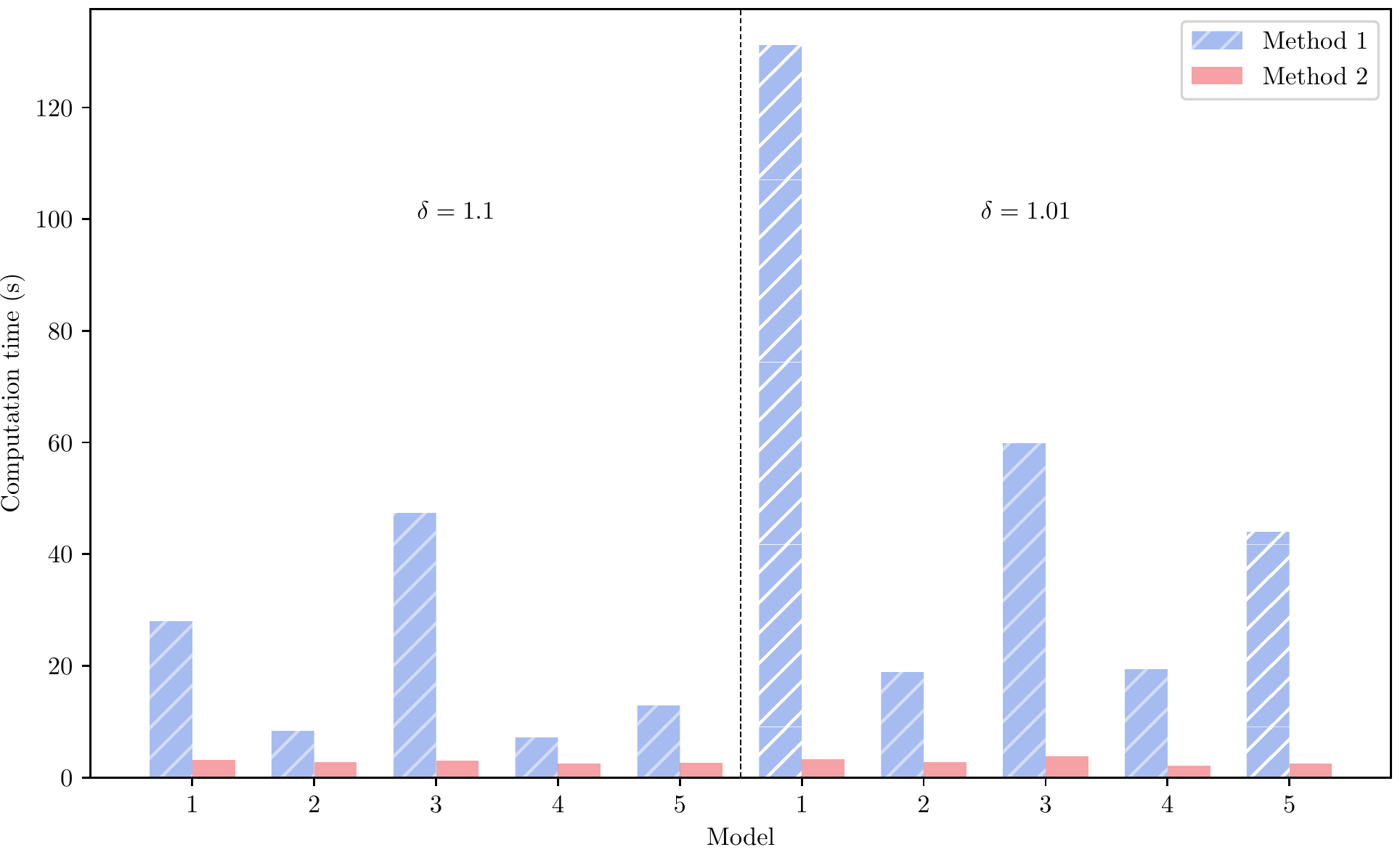}
  \caption{Computation time comparison for the two methods.}
  \label{fig:speed}
\end{figure}

The main advantage of this algorithm is that, for any chosen number of grid
points, the number of minimization operations required is fixed, and we can
improve the accuracy of this algorithm by increasing the number of grid
points. For the iteration method, however, the rate of convergence is
different for different model settings and to achieve the same accuracy it
usually requires longer computation time.

In Figure \ref{fig:speed}, we plot the computation time\footnote{The
  computations were conducted on a XPS 13 9360 laptop with i7-7500U CPU. The
  program only utilizes a single core.} of successive iterations of $T$ with
$p_0 = c$ (method 1) and Algorithm 1 (method 2) for ten different model
settings when the number of grid points is set to be $m = 1000$. The first
and last five models are the same\footnote{The cost function $c$ and
  additive transaction cost function $g$ for the five models are: (1) $c(s)
  = e^{10s} - 1$, $g(k) = k-1$; (2) $c(s) = e^{s} - 1$, $g(k) = 0.01(k-1)$;
  (3) $c(s) = e^{s^2} - 1$, $g(k) = 0.01(k-1)$; (4) $c(s) = s^2 + s$, $g(k)
  = 0.01(k-1)$; (5) $c(s) = e^{s} + s^2 - 1$, $g(k) = 0.05(k-1)$.} except
$\delta=1.1$ for the former and $\delta=1.01$ for the latter. In each model,
we also compute an accurate price function using Algorithm 1 with a very
large number of grid points ($m = 50000$) and compare it with results from
both methods when $m = 1000$. We find that the error from method 2 is
comparable or smaller than that from method 1 in each model. The algorithm
achieves more accurate results at a much faster speed. As we can see in
Figure~\ref{fig:speed}, method 2 completes the computation in around 3
seconds in each model while the computation time of method 1 ranges from 7
seconds to more than 2 minutes. The speed difference is especially drastic
when $\delta$ is close to 1, since it takes $T$ more iterations to converge
with smaller $\delta$ but the number of operations for the algorithm is
fixed. In model 1 with $\delta = 1.01$, the algorithm is 40 times faster
than successive iterations of $T$!

\section{Stochastic Choices}
\label{sec:st}

So far we have discussed the case in which each firm can choose the optimal
number of upstream partners according to (\ref{eq:prob}). In reality,
however, firms usually face uncertainty when choosing their partners. The
result is that some firms might choose fewer or more partners than what is
optimal. For instance, a firm might not be able to choose a certain number
of upstream partners due to regulation or failure to arrive at agreements
with potential partners. Conversely, the upstream partners of a firm might
experience supply shocks and fail to meet production requirements, causing
it to sign more partners than what is optimal and bear more transaction
costs. In this section, we model this scenario and incorporate uncertainty
into each firm's optimization problem.

We assume that each firm chooses an amount of ``search effort'' $\lambda$ and
the resulting number of upstream partners follows a Poisson
distribution\footnote{Note that in the usual sense, if a random variable $X$
  follows the Poisson distribution, $X$ takes values in nonnegative
  integers. Here we shift the probability function so that $k$ starts from 1.}
with parameter $\lambda$ that starts from $k=1$. In other words, the
probability of having $k$ partners is
\begin{equation*}
  f(k;\lambda) = \frac{\lambda^{k-1} e^{-\lambda}}{(k-1)!}
\end{equation*}
when $\lambda > 0$. We also assume that when $\lambda=0$, Prob$(n = 1) = 1$,
that is, each firm can always choose to have only one upstream partner with
certainty. For example, if a firm chooses to exert effort $\lambda = 2.5$, the
probabilities of it ending up with 1, 2, 3, 4, 5 partners are, respectively:
0.08, 0.2, 0.26, 0.21, 0.13. One characteristic of the Poisson distribution is
that both its mean and variance increase with $\lambda$, which makes it
suitable for our model since the more partners a firm aims for, the more
uncertainty there will be in the contracting process.

Hence, a firm at stage $s$ solves the following problem:
\begin{equation}
  \label{eq:prob2}
  \min_{\substack{t\leq s\\ \lambda\geq 0}}\left\{ c(s - t) +
    \E_k^\lambda \left[g(k) + \delta k p(t/k) \right] \right\}
\end{equation}
where $\E_k^\lambda$ stands for taking expectation of $k$ under the Poisson
distribution with parameter $\lambda$. Specifically,
\begin{equation*}
  \E_k^\lambda \left[g(k) + \delta k p(t/k) \right] = \sum_{k=1}^\infty
  \left[g(k) + \delta k p(t/k) \right] f(k;\lambda).
\end{equation*}
Similar to Section \ref{sec:eq}, we can define another operator $\tilde{T}:
C(X) \to R(X)$ by
\begin{equation}
  \label{eq:T2}
  \tilde{T}p(s):=   \min_{\substack{t\leq s\\ \lambda\geq 0}}\left\{ c(s - t) +
    \E_k^\lambda \left[g(k) + \delta k p(t/k) \right] \right\}.
\end{equation}
As will be shown in \ref{app:3}, all of the above results still apply
in the stochastic case and we summarize them in the following theorem.
\begin{thm}
  \label{main:st}
  Let $u_0(s) = c'(0)s$, $v_0(s) = c(s)$. If Assumption \ref{c} and \ref{g}
  hold, then the operator $\tilde{T}$ has a unique fixed point $\tilde{p}^*$
  in $[u_0, v_0]$ and $\tilde{T}^np \to \tilde{p}^*$ for any $p\in [u_0,
  v_0]$. Furthermore, $\tilde{p}_n$ from Algorithm 1 converges to
  $\tilde{p}^*$ uniformly.
\end{thm}

By Theorem \ref{main:st}, there exists a unique equilibrium price function
$\tilde{p}^*$ and we can compute it either by successive evaluation of
$\tilde{T}$ or by Algorithm 1. The algorithm is particularly useful here since
it now takes much longer time to complete one minimization operation with
firms choosing continuous values of $\lambda$ instead of discreet values of
$k$.

Similarly, there exist minimizers $t^*$ and $\lambda^*$ so that firm at any
stage $s$ has an optimal choice $t^*(s)$ and $\lambda^*(s)$. With the optimal
choice functions, we can compute an equilibrium firm allocation recursively as
in \cite{kikuchi2018span}. Specifically, we start at the most downstream
firm at $s=1$ and compute its optimal choices $t^*$ and $\lambda^*$. Next, we
pick a realization of $k$ according to the Poisson distribution with parameter
$\lambda^*$ and repeat the process for each of its upstream firm at $s' =
t^*/k$. The whole process ends when all the most upstream firms choose to
carry out the remaining production process by themselves. Note that due to the
stochastic nature of this model, each simulation will give a different firm
allocation.

\begin{figure}[tb!]
  \centering
  \begin{subfigure}{0.5\textwidth}
    \centering
    \includegraphics[height=0.88\textwidth]{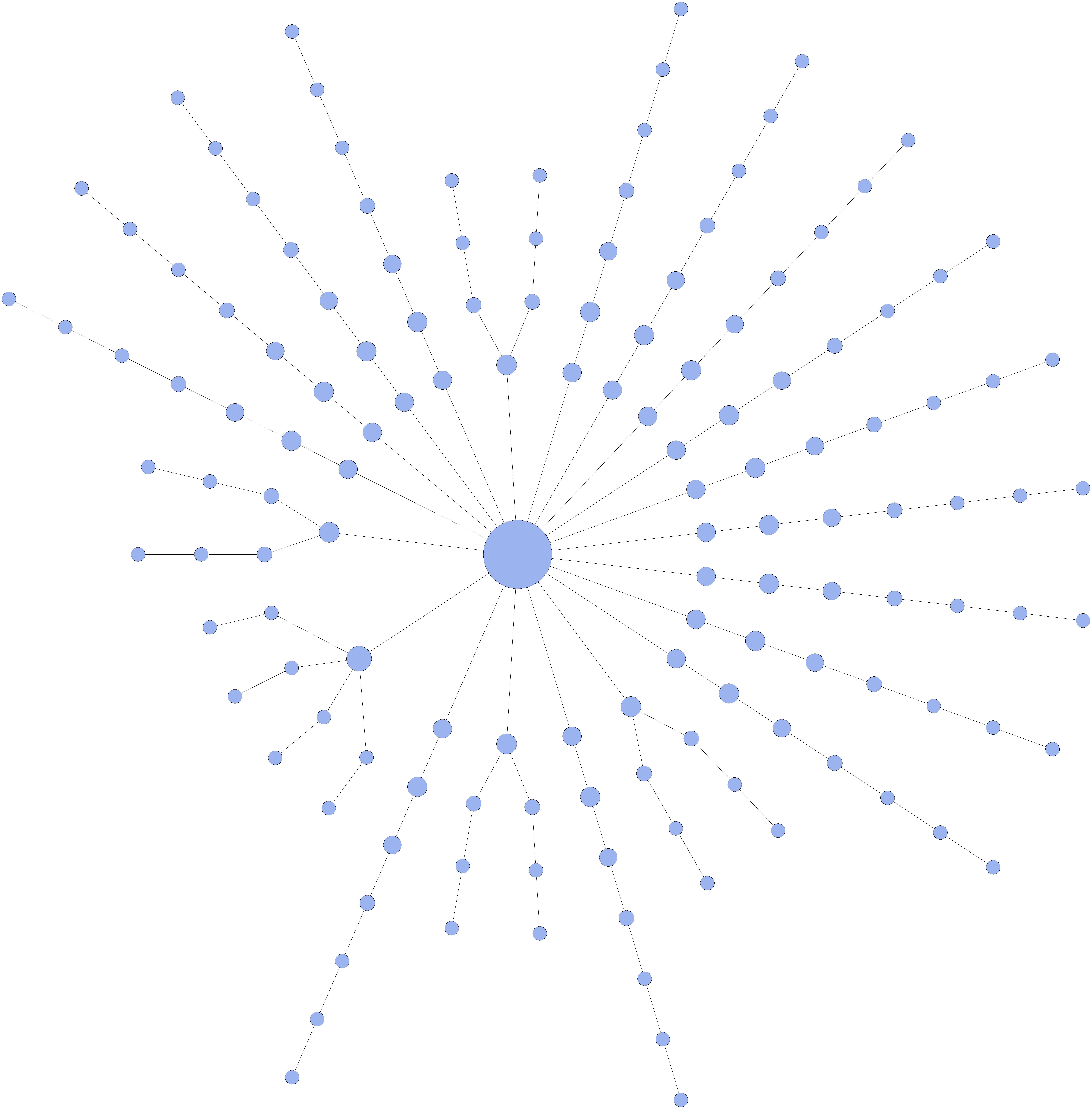}
    \caption{$\beta = 0.0005$, $\delta = 1.05$, $\theta = 1.2$}
  \end{subfigure}%
  \begin{subfigure}{0.5\textwidth}
    \centering
    \includegraphics[height=0.88\textwidth]{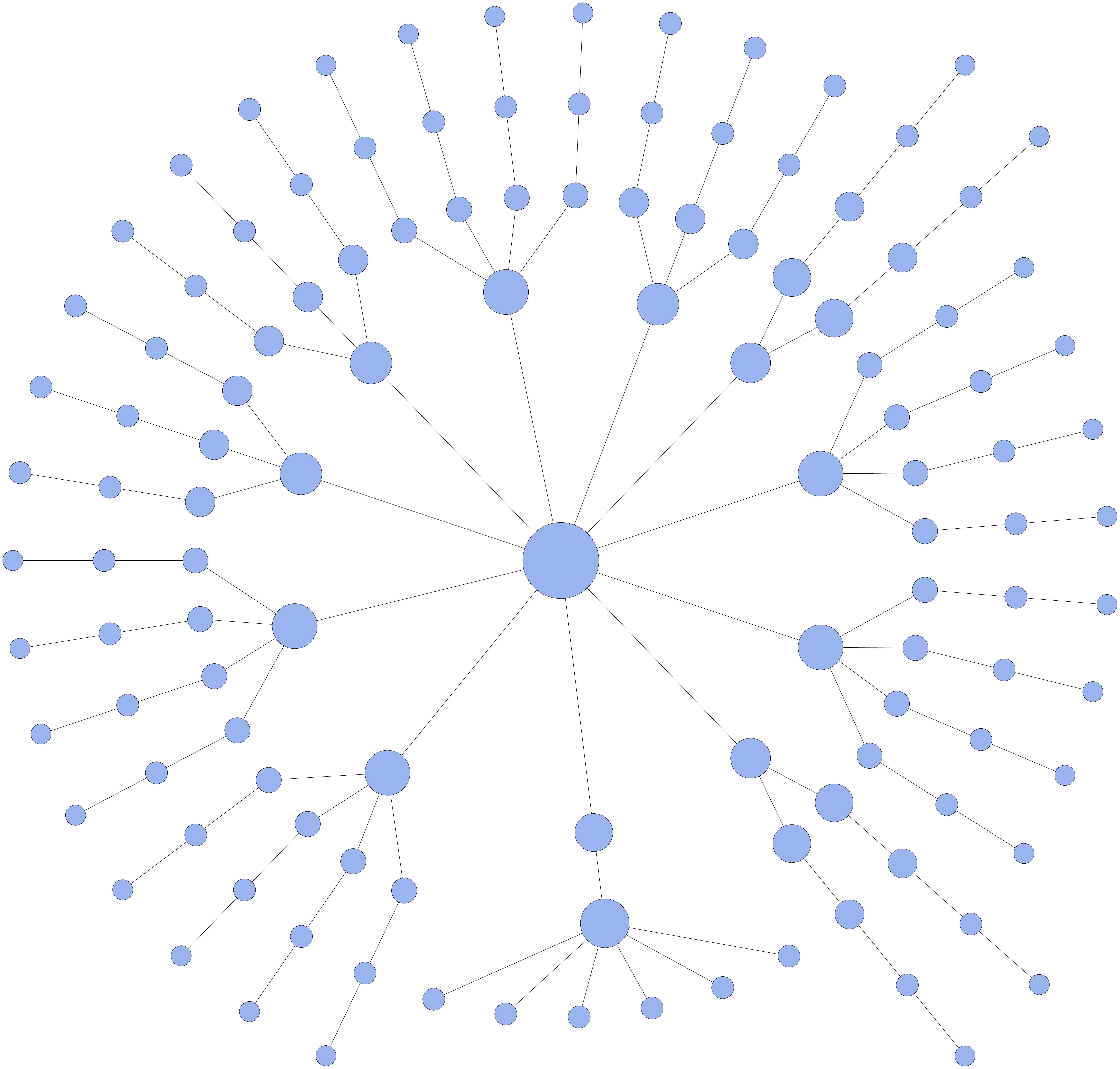}
    \caption{$\beta = 0.0005$, $\delta = 1.1$, $\theta = 1.2$}
  \end{subfigure}
  \begin{subfigure}{0.5\textwidth}
    \centering
    \includegraphics[height=0.88\textwidth]{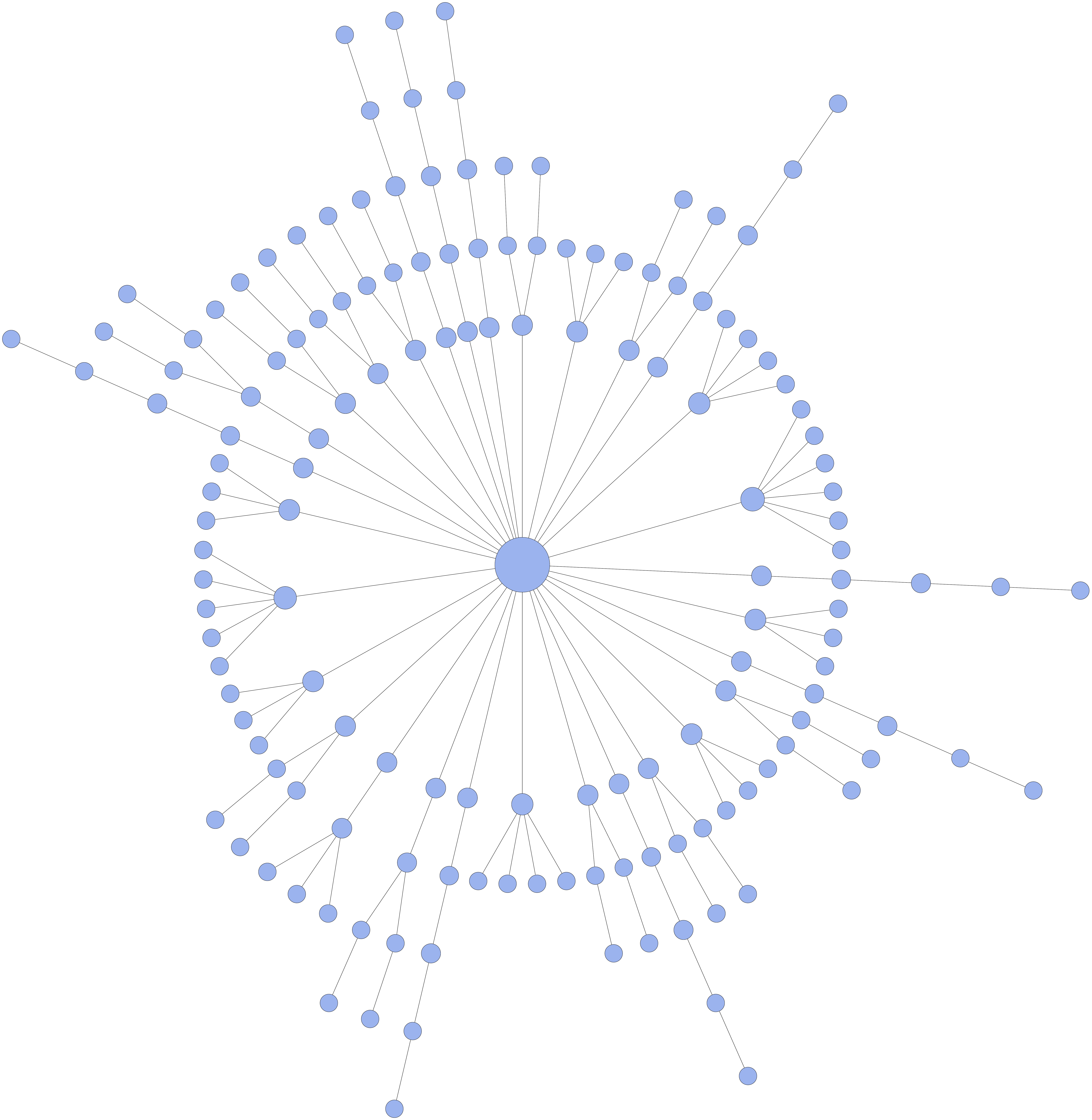}
    \caption{$\beta = 0.0001$, $\delta = 1.05$, $\theta = 1.2$}
  \end{subfigure}%
  \begin{subfigure}{0.5\textwidth}
    \centering
    \includegraphics[height=0.88\textwidth]{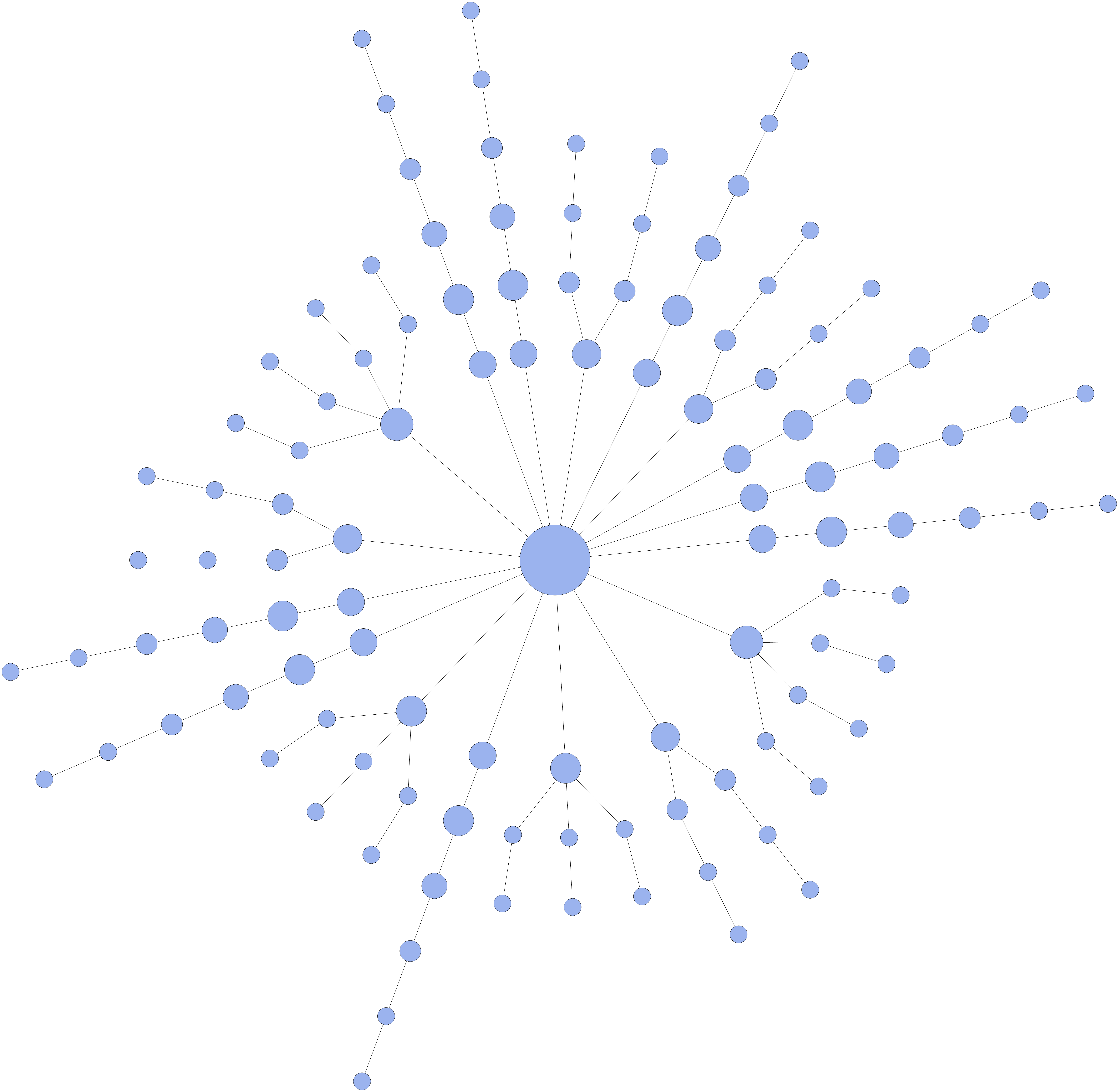}
    \caption{$\beta = 0.0005$, $\delta = 1.05$, $\theta = 1.15$}
  \end{subfigure}
  \caption{Production networks with stochastic choices of upstream partners}
  \label{fig:network}
\end{figure}

In Figure~\ref{fig:network}, we plot some production networks for different
model parameterizations using the above approach. Each node represents a
firm and the one at the center is the firm at $s = 1$. The size of each node
is proportionate to the size\footnote{Here the firm size is calculated using
  its value added $c(s - t^*) + g(k)$ where $k$ is a realization of the
  Poisson distribution with parameter $\lambda^*$.} of the corresponding
firm. The cost function is set to be $c(s) = s^\theta$ and the additive
transaction cost is $g(k) = \beta (k-1)^{1.5}$. Compared with production
networks in \cite{kikuchi2018span}, the graphs here are no longer
symmetric since even firms on the same layer can have different realized
numbers of upstream partners and thus different firm sizes. The prediction
that downstream firms are larger and tend to have more subcontractors, on
the other hand, is also valid in our networks.

Comparing (a) and (b), an increase in transaction cost makes firms in (b)
outsource less and produce more in-house, resulting in fewer layers in the
production network. Similarly, comparing (c) with (a), a decrease in
additive transaction costs encourages firms at each level to find more
subcontractors. The results are more but smaller firms at each level and
fewer layers in the network. Comparing (d) with (a), the difference is a
decrease in curvature of the cost function $c$, which makes outsourcing less
appealing. The firms in (d) tend to produce more in-house, resulting in a
production network of fewer layers.

\section{Conclusion}
\label{sec:con}

In this paper, we extend the production chain model of
\cite{kikuchi2018span} to more realistic settings, in which each firm
can have multiple upstream partners and face uncertainty in the contracting
process. We prove the uniqueness of equilibrium price for these extensions
and propose a fast algorithm for computing the price function that is
guaranteed to converge.

The key to proving uniqueness of equilibrium price in this model is the
theory of monotone concave operators, which gives sufficient conditions for
existence, uniqueness, and convergence. This theory has been proven useful
in finding equilibria in a range of economic models as mentioned in the
introduction and can potentially be applied to more problems where
contraction property is hard to establish.

Our model also has some predictions regarding the shape of production
networks and the size distribution of firms. In our extended model with
uncertainty, we generate a series of production networks
(Figure~\ref{fig:network}) under different model settings. A notable
observation from this exercise is that increasing the proportionate
transaction cost $\delta$ or decreasing the additive transaction cost $g$
will reduce the number of layers in a network. In the former case, the cost
of market transactions increases; this encourages vertical integration and
hence leads to larger firms along each chain. In the latter case, the cost
of maintaining multiple partners decreases; this discourages lateral
integration and leads to more firms in each layer. This prediction can
potentially be tested with suitable choice of proxies for $\delta$ and $g$.

Another observation is that different model settings lead to different size
distributions of firms. For example, smaller $\delta$ seems to lead to more
extreme differences in firm sizes as shown in the comparison between (a) and
(b) in Figure~\ref{fig:network}. The underlying mechanism is unclear in our
model, which provides a possible channel for future research.

A notable feature of our model is that firms are ex-ante identical but
ex-post heterogeneous in equilibrium in terms of sizes, positions in a
network, and number of subcontractors. However, the cost function $c$ and
transaction costs $\delta$ and $g$ are assumed to be fixed throughout this
paper. Introducing heterogeneity into these costs might offer richer
implications for firm distribution and industry policies. We also leave this
possibility for future research.

\appendix
\renewcommand{\thethm}{\Alph{section}.\arabic{thm}}
\section{Proofs from Section \ref{sec:eq}}
\label{app:1}

Let $U = \N \times [0, 1]$ equipped with the Euclidean metric in $\R^2$ and
$X$ be equipped with the Euclidean metric in $\R$.  To simplify notation, we
can write $T$ as
\begin{equation*}
  Tp(s) = \min_{(k, t) \in \Theta(s)} f_p(s, k, t)
\end{equation*}
where $\Theta:X \to U$ is a correspondence defined by $\Theta(s) = \N \times
[0, s]$, and $f_p(s, k, t) = c(s - t) + g(k) + \delta k p(t/k)$.

\begin{lemma}
  \label{berge}
  $Tp \in C([0, 1])$ for all $p\in C([0, 1])$.
\end{lemma}
\begin{proof}
  We use Berge's theorem to prove continuity. By Assumption \ref{g}, we can
  restrict $\Theta$ to be $\Theta(s) = \{1, 2, \ldots, \bar{k}\} \times [0,
  s]$ for some large $\bar{k}\in \N$. Then $\Theta$ is compact-valued.

  To see $\Theta$ is upper hemicontinuous, note $\Theta(s)$ is closed for
  all $s\in X$. Since the graph of $\Theta$ is also closed, by the Closed
  Graph Theorem \citep[see, e.g.,][p. 565]{aliprantis2006infinite}, $\Theta$
  is upper hemicontinuous on $X$.


  To check for lower hemicontinuity, fix $s \in X$. Let $V$ be any open set
  intersecting $\Theta(s) = \{1, 2, \ldots, \bar{k}\} \times [0, s]$. Then it
  is easy to see that we can find a small $\epsilon > 0$ such that $\Theta(s')
  \cap V \neq \emptyset$ for all $s' \in [s-\epsilon, s+\epsilon]$. Hence
  $\Theta$ is lower hemicontinuous on $X$.

  Because $p\in C([0, 1])$, $f_p$ is jointly continuous in its three
  arguments. By Berge's theorem, $Tp$ is continuous on $X$.
\end{proof}

Note that by Berge's theorem, the minimizers $t^*$ and $k^*$ exist and are
upper hemicontinuous.


\begin{lemma}
  \label{T}
  T is increasing and concave.
\end{lemma}
\begin{proof}
  It is apparent that $T$ is increasing. To see $T$ is concave, let $p, q\in
  C([0, 1])$ and $\alpha \in (0, 1)$. Then we have
  \begin{align*}
    \alpha Tp(s) + (1-\alpha) Tq(s)
    &= \min_{(k, t)\in \Theta(s)}\alpha f_p(s, k, t) + \min_{(k, t)\in
      \Theta(s)}(1 - \alpha) f_q(s, k, t)\\
    &\leq \min_{(k, t)\in \Theta(s)} \left\{\alpha f_p(s, k, t) + (1-\alpha)
      f_q(s, k, t)\right\}\\
    &= \min_{(k, t)\in \Theta(s)} \left\{c(s - t) + g(k) + \delta k\left[\alpha
      p(t/k) + (1-\alpha)q(t/k)\right] \right\}\\
    &= \min_{(k, t)\in \Theta(s)} f_{\alpha p + (1-\alpha) q}(s, k, t)\\
    &= T\left[\alpha p + (1-\alpha) q\right](s)
  \end{align*}
  which completes the proof.
\end{proof}

\begin{lemma}
  \label{inc_con}
  $Tu_0 \geq u_0 + \epsilon(v_0 - u_0)$ for some $\epsilon \in (0, 1)$.
\end{lemma}
\begin{proof}
  Define $\bar{s} := \max\{0 \leq s \leq 1: c'(s) \leq \delta c'(0)\}$. Then
  we have
  \begin{align*}
    Tu_0(s) &= \min_{(k, t)\in \Theta(s)} f_{u_0}(s, k, t)\\
            &= \min_{(k, t)\in \Theta(s)} \left\{ c(s-t) + g(k) + \delta
              c'(0)t \right\} \\
            &= \min_{t\leq s} \left\{ c(s -t) + \delta c'(0) t\right\}\\
            &=
              \begin{cases}
                c(\bar{s}) + \delta c'(0) (s - \bar{s}), \quad &\text{if }
                s\geq \bar{s}\\
                c(s), \quad &\text{if } s<\bar{s}
              \end{cases}
  \end{align*}
  Since $Tu_0(s) > u_0(s)$ for all $s$ except at 0, we can find $\epsilon \in
  (0, 1)$ such that $Tu_0 \geq u_0 + \epsilon(v_0 - u_0)$.
\end{proof}

\begin{lemma}
  \label{v}
  $Tv_0 \leq v_0$.
\end{lemma}
\begin{proof}
  Choose $k = 1$ and $t = 0$. We have $Tv_0(s) \leq c(s -0) + g(1) + \delta
  c(0) = c(s) = v_0(s)$.
\end{proof}

\begin{proof}[Proof of Theorem \ref{main}]
  Since $P = \left\{f\in C(X): f(x) \geq 0 \text{ for all } x\in X\right\}$ is
  a normal cone, the theorem follows from the previous lemmas and Theorem
  \ref{thm:du}.
\end{proof}

\begin{proof}[Proof of Proposition \ref{p_inc}]
  We first show that $T$ maps a strictly increasing function to a strictly
  increasing function. Suppose $p\in [u_0, v_0]$ and is strictly
  increasing. Pick any $s_1, s_2 \in [0, 1]$ with $s_1 < s_2$. Let $t^*$ and
  $k^*$ be the minimizers of $T$. To simplify notation, let $t_1 \in
  t^*(s_1)$, $t_2 \in t^*(s_2)$, $k_1 \in k^*(s_1)$, and $k_2 \in
  k^*(s_2)$. If $t_2 \leq s_1$, then we have
  \begin{align*}
    Tp(s_2) &= c(s_2 - t_2) + g(k_2) + \delta k_2 p(t_2/k_2)\\
            &> c(s_1 - t_2) + g(k_2) + \delta k_2 p(t_2/k_2)\\
            &\geq Tp(s_1).
  \end{align*}
  If $s_1 < t_2 \leq s_2$, then $t_2 + s_1 - s_2 \leq s_1$. Since $p$ is
  strictly increasing, we have
  \begin{align*}
    Tp(s_2) &= c\left(s_1 - (t_2 + s_1 - s_2)\right) + g(k_2) + \delta k_2
              p(t_2/k_2)\\ 
            &> c\left(s_1 - (t_2 + s_1 - s_2)\right) + g(k_2) + \delta k_2
              p\left((t_2 + s_1 - s_2)/k_2\right)\\
            &\geq Tp(s_1).
  \end{align*}
  
  Since $c\in [u_0, v_0]$, by Theorem \ref{main}, $T^n c \to p^*$ as $n\to
  \infty$. Furthermore, since $c$ is strictly increasing, it follows from the
  above result that $p^*$ is strictly increasing.
\end{proof}

\begin{proof}[Proof of Proposition \ref{comp}]
  If $\delta_a \leq \delta_b$, then $T_ap \leq T_b p$ for any $p\in[u_0,
  v_0]$. Since $T$ is increasing by Lemma~\ref{T}, we have $T_a^n p \leq T_b^n
  p$ for any $p\in [u_0, v_0]$ and any $n\in\N$. Then by Theorem~\ref{main},
  $p_a^* \leq p_b^*$. The same arguments applies if $g_a \leq g_b$.
\end{proof}

\section{Proof of Theorem \ref{alg_conv}}
\label{app:2}

\begin{lemma}
  \label{inc}
  The function $p_n$ is increasing for every $n$.
\end{lemma}
\begin{proof}
  As $p_n$ is piecewise linear, we shall prove it by induction. Since $p_n(0)
  = 0$ and $p_n(h_n) = c(h_n)$, $p_n$ is increasing on $[0, h_n]$. Suppose it
  is increasing on $[0, s]$ for some $s = h_n, 2h_n, \ldots, 1-h_n$, then we
  have
  \begin{align*}
    p_n(s+h_n) &= \min_{t\leq s,\;k\in\N} \left\{c(s + h_n -t) + g(k) + \delta
                 k p_n(t/k) \right \}\\
               &= c(s + h_n -t^*) + g(k^*) + \delta k^* p_n(t^*/k^*)
  \end{align*}
  where $t^*$ and $k^*$ are the minimizers. If $t^* \leq s-h_n$, it follows
  from the monotonicity of $c$ that
  \begin{align*}
    p_n(s+h_n) &\geq c(s - t^*) + g(k^*) + \delta k^* p_n(t^*/k^*)\\
               &\geq \min_{t\leq s-h_n,\;k\in\N} \left\{c(s - t) + g(k) +
                 \delta k p_n(t/k) \right \}\\
               &= p_n(s).
  \end{align*}
  If $t^* \in (s-h_n, s]$, then $s + h_n - t^* \geq h_n$. Because $p_n$ is
  increasing on $[0, s]$, we have
  \begin{align*}
    p_n(s+h_n) &\geq c[s - (s - h_n)] + g(k^*) + \delta k^* p_n[(s-h_n)/k^*]\\
               &\geq \min_{t\leq s-h_n,\;k\in\N} \left\{c(s - t) + g(k) +
                 \delta k p_n(t/k) \right \}\\
               &= p_n(s),
  \end{align*}
  which completes the proof.
\end{proof}

\begin{lemma}
  \label{equicon}
  The sequence $\{p_n\}_{n=1}^\infty$ is uniformly bounded and equicontinuous.
\end{lemma}
\begin{proof}
  To see $\{p_n\}$ is uniformly bounded, note that for each $n$,
  \begin{align*}
    p_n(s + h_n) &= \min_{t\leq s,\;k\in\N} \left\{c(s + h_n - t) + g(k) + \delta
                   k p_n(t/k) \right \}\\
                 &\leq c(s + h_n) + g(1) + \delta p_n(0)\\
                 &= c(s + h_n) \leq c(1)
  \end{align*}
  for all $s = 0, h_n, \ldots, 1 - h_n$. 

  Due to Lemma \ref{inc}, to see $\{p_n\}$ is equicontinuous, it suffices to
  show that there exists $K > 0$ such that $p_n(s + h_n) - p_n(s) \leq Kh_n$
  for all $n\in \N$ and all $s = 0, h_n, 2h_n, \ldots, 1 - h_n$. Fix such $n$
  and $s$. If $s = 0$, $p_n(h_n) - p_n(0) = c(h_n) \leq c'(1) h_n$. If $s \geq
  h_n$, denote the minimizers in the definition of $p_n(s)$ by $t^*$ and
  $k^*$, i.e.,
  \begin{align*}
    p_n(s) &= \min_{t\leq s - h_n,\;k\in\N} \left\{c(s - t) + g(k) + \delta k
           p_n(t/k) \right \}\\
         &= c(s - t^*) + g(k^*) + \delta k^* p_n(t^*/k^*).
  \end{align*}
  Since $t^* \leq s$, it follows that
  \begin{align*}
    p_n(s + h_n) &= \min_{t\leq s,\;k\in\N} \left\{c(s + h_n -t) + g(k) + \delta
                   k p_n(t/k) \right \}\\
                 &\leq c(s + h_n - t^*) + g(k^*) + \delta k^* p_n(t^*/k^*).
  \end{align*}
  Hence,
  \begin{align*}
    p_n(s + h_n) - p_n(s) &\leq c(s + h_n - t^*) + c(s - t^*)\\
                          &\leq c'(1) h_n,
  \end{align*}
  which completes the proof.
\end{proof}

\begin{lemma}
  \label{conv}
  There exists a uniformly convergent subsequence of $\{p_n\}$. Furthermore,
  every uniformly convergent subsequence of $\{p_n\}$ converges to a fixed
  point of $T$.
\end{lemma}
\begin{proof}
  Lemma \ref{equicon} and the Arzel\`a-Ascoli theorem imply that ${p_n}$ has
  a uniformly convergent subsequence. To simplify notation, let $\{p_n\}$ be
  such a subsequence and converge uniformly to $\bar{p}$. Because $p_n$ are
  continuous, $\bar{p}$ is continuous. By Berge's theorem,
  \begin{equation*}
    T\bar{p}(s) = \min_{t\leq s,\;k\in\N} \left\{ c(s-t) + g(k) + \delta k
      \bar{p}(t/k) \right\}
  \end{equation*}
  is also continuous. To see $\bar{p}$ is a fixed point of $T$, it is
  sufficient to show that $\bar{p}$ and $T\bar{p}$ agree on the dyadic
  rationals $\cup_nG_n$, i.e.,
  \begin{equation*}
    \lim_{n\to\infty} \min_{\substack{t\leq s - h_n\\ k\in\N}} \left\{c(s - t)
      + g(k) + \delta k p_n(t/k) \right \} = \min_{t\leq s,\;k\in\N} \left\{
      c(s-t) + g(k) + \delta k \bar{p}(t/k) \right\} 
  \end{equation*}
  for every $s \in \cup_nG_n$.

  Fix $\epsilon > 0$. Since $p_n \to \bar{p}$ uniformly, there exists $N_1 \in
  \N$ such that $n > N_1$ implies that
  \begin{equation*}
    p_n(x) > \bar{p}(x) - \epsilon/(\delta\bar{k})
  \end{equation*}
  for all $x\in [0, 1]$ where $\bar{k}$ is the upper bound on the possible
  values of $k$. It follows that for $n > N_1$ we have
  \begin{align*}
    \min_{\substack{t\leq s - h_n\\ k\in\N}} \left\{c(s - t) + g(k) + \delta k
    p_n(t/k) \right \}
    &> \min_{\substack{t\leq s - h_n\\ k\in\N}} \left\{c(s - t) + g(k) +
    \delta k \bar{p}(t/k) \right \} - \epsilon\\
    &\geq \min_{\substack{t\leq s\\ k\in\N}} \left\{c(s - t) + g(k) +
    \delta k \bar{p}(t/k) \right \} - \epsilon.
  \end{align*}
  Therefore,
  \begin{equation*}
    \lim_{n\to\infty} \min_{\substack{t\leq s - h_n\\ k\in\N}} \left\{c(s - t)
      + g(k) + \delta k p_n(t/k) \right \} \geq \min_{t\leq s,\;k\in\N} \left\{
      c(s-t) + g(k) + \delta k \bar{p}(t/k) \right\}.
  \end{equation*}

  For the other direction, there exists $N_2\in \N$ such that $n > N_2$
  implies that
  \begin{equation*}
    p_n(x) < \bar{p}(x) + \epsilon/(2\delta\bar{k})
  \end{equation*}
  for all $x\in[0, 1]$. Then for $n > N_2$ we have
  \begin{equation*}
    \min_{\substack{t\leq s - h_n\\ k\in\N}} \left\{c(s - t) + g(k) + \delta k
      p_n(t/k) \right \} < \min_{\substack{t\leq s - h_n\\ k\in\N}} \left\{c(s
      - t) + g(k) + \delta k \bar{p}(t/k) \right \} + \epsilon/2.
  \end{equation*}
  Since $c, g, \bar{p}$ are continuous and $h_n \to 0$, we can choose $N_3$
  such that $n > N_3$ implies that
  \begin{equation*}
    \min_{\substack{t\leq s - h_n\\ k\in\N}} \left\{c(s - t) + g(k) + \delta k
      \bar{p}(t/k) \right \} < \min_{\substack{t\leq s\\ k\in\N}}
    \left\{c(s - t) + g(k) + \delta k \bar{p}(t/k) \right \} + \epsilon/2.
  \end{equation*}
  Hence, for $n > \max\{N_2, N_3\}$ we have 
  \begin{equation*}
    \min_{\substack{t\leq s - h_n\\ k\in\N}} \left\{c(s - t) + g(k) + \delta k
      p_n(t/k) \right \} < \min_{\substack{t\leq s\\ k\in\N}} \left\{c(s -
      t) + g(k) + \delta k \bar{p}(t/k) \right \} + \epsilon. 
  \end{equation*}
  This implies
  \begin{equation*}
    \lim_{n\to\infty} \min_{\substack{t\leq s - h_n\\ k\in\N}} \left\{c(s - t)
      + g(k) + \delta k p_n(t/k) \right \} \leq \min_{t\leq s,\;k\in\N} \left\{
      c(s-t) + g(k) + \delta k \bar{p}(t/k) \right\}.
  \end{equation*}
  Therefore, $\bar{p} = T\bar{p}$.
\end{proof}

\begin{lemma}
  \label{pstar}
  Every uniformly convergent subsequence of $\{p_n\}$ converges to $p^*$.
\end{lemma}
\begin{proof}
  Let $\{p_n\}$ be the subsequence that converges uniformly to $\bar{p}$. By
  Theorem \ref{main}, to see $\bar{p} = p^*$, it suffices to show that
  $\bar{p}$ is continuous and $c'(0)x \leq \bar{p}(x) \leq c(x)$ for all
  $x\in[0, 1]$. Continuity is satisfied by the fact that each $p_n$ is
  continuous and $p_n\to \bar{p}$ uniformly. To show the second one, we again
  prove this holds on $\cup_nG_n$, and it is sufficient to show that $c'(0)s
  \leq p_n(s) \leq c(s)$ for all $s \in G_n$ and all $n\in\N$. It is apparent
  that $p_n(s) \leq c(s)$ (choose $t = 0$ and $k = 1$). We show $p_n(s) \geq
  c'(0)s$ by induction. Suppose $p_n(x) \geq c'(0)x$ for all $x \leq s$. Then
  we have
  \begin{align*}
    p_n(s+h_n) &= \min_{t\leq s,\;k\in\N} \left\{c(s+h_n-t) + g(k) + \delta k
                 p_n(t/k) \right \}\\
               &\geq \min_{t\leq s,\;k\in\N} \left\{c'(0)(s+h_n-t) + g(k) +
                 \delta c'(0) t \right \}\\
               &=  \min_{t\leq s} \left\{c'(0)(s+h_n-t + \delta t)
                 \right \}\\
               &= c'(0)(s + h_n).
  \end{align*}
  Since $p_n(0) = 0 \geq c'(0)\cdot 0$, it follows that $p_n(s) \geq
  c'(0)s$. This concludes the proof.
\end{proof}

\section{Proof of Theorem~\ref{main:st}}
\label{app:3}

Similar to \ref{app:1}, we can write the operator $\tilde{T}$ in
(\ref{eq:T2}) as
\begin{equation*}
  \tilde{T}p(s) = \min_{(\lambda, t)\in \tilde{\Theta}(s)}\left\{ c(s - t) +
    \E_k^\lambda \left[g(k) + \delta k p(t/k) \right] \right\}
\end{equation*}
where $\tilde{\Theta}(s) = [0, \infty) \times [0, s]$. Upon close inspection,
all of the above lemmas still hold for $\tilde{T}$ if we can restrict
$\tilde{\Theta}(s)$ to be a compact set. To be more specific, Lemma~\ref{T}
and \ref{inc} can be proved in the exact same way; Lemma~\ref{inc_con},
\ref{v}, \ref{equicon}, and \ref{pstar} hold since each firm can choose $k=1$
with probability 1; Lemma~\ref{conv} and \ref{berge} need the compactness of
$\tilde{\Theta}(s)$. To avoid redundancy, we omit the proofs and shall only
show that there exists an upper bound on the choice set of $\lambda$.

Let $\nu$ be the median of the Poisson distribution and denote the ceiling
of $\nu$ (i.e., the least integer greater than or equal to $\nu$) by
$\bar{\nu}$. Then we have
\begin{equation*}
  \sum_{k = \bar{\nu}}^\infty f(k; \lambda) \geq \frac{1}{2}
\end{equation*}
by definition. It follows that the expectation of $g(k)$
\begin{align*}
  \E^\lambda_k g(k) &= \sum_{k=1}^\infty g(k) f(k; \lambda)\\
                    &\geq \sum_{k=\bar{\nu}}^\infty g(k) f(k; \lambda)\\
                    &\geq g(\bar{\nu})\sum_{k=\bar{\nu}}^\infty f(k;
                      \lambda)\\
                    &\geq \frac{1}{2} g(\bar{\nu})
\end{align*}
where the second inequality follows from
Assumption~\ref{g}. \cite{choi1994medians} gives bounds\footnote{Since in
  our model $k$ starts from 1, we write $\nu - 1$ in the inequality.} for the
median of the Poisson distribution:
\begin{equation*}
  \lambda - \ln 2 \leq \nu - 1 < \lambda + \frac{1}{3}.
\end{equation*}
So we have
\begin{equation*}
  \E^\lambda_k g(k) \geq \frac{1}{2}g(\bar{\nu}) \geq \frac{1}{2}g(\nu) \geq
  \frac{1}{2} g(\lambda - \ln 2 + 1).
\end{equation*}
Therefore, we can find $\bar{\lambda}$ such that $\E^\lambda_k g(k) \geq c(1)$
for all $\lambda \geq \bar{\lambda}$ and hence $\Theta(s)$ is essentially $[0,
\bar{\lambda}] \times [0, s]$ which is a compact set.
\bibliography{bib}

\begin{thebibliography}{48}
\expandafter\ifx\csname natexlab\endcsname\relax\def\natexlab#1{#1}\fi
\providecommand{\url}[1]{\texttt{#1}}
\providecommand{\href}[2]{#2}
\providecommand{\path}[1]{#1}
\providecommand{\DOIprefix}{doi:}
\providecommand{\ArXivprefix}{arXiv:}
\providecommand{\URLprefix}{URL: }
\providecommand{\Pubmedprefix}{pmid:}
\providecommand{\doi}[1]{\href{http://dx.doi.org/#1}{\path{#1}}}
\providecommand{\Pubmed}[1]{\href{pmid:#1}{\path{#1}}}
\providecommand{\bibinfo}[2]{#2}
\ifx\xfnm\relax \def\xfnm[#1]{\unskip,\space#1}\fi
\bibitem[{Acemoglu et~al.(2012)Acemoglu, Carvalho, Ozdaglar and
  Tahbaz-Salehi}]{acemoglu2012network}
\bibinfo{author}{Acemoglu, D.}, \bibinfo{author}{Carvalho, V.M.},
  \bibinfo{author}{Ozdaglar, A.}, \bibinfo{author}{Tahbaz-Salehi, A.},
  \bibinfo{year}{2012}.
\newblock \bibinfo{title}{The network origins of aggregate fluctuations}.
\newblock \bibinfo{journal}{Econometrica} \bibinfo{volume}{80},
  \bibinfo{pages}{1977--2016}.
\bibitem[{Acemoglu et~al.(2015a)Acemoglu, Ozdaglar and
  Tahbaz-Salehi}]{acemoglu2015networks}
\bibinfo{author}{Acemoglu, D.}, \bibinfo{author}{Ozdaglar, A.},
  \bibinfo{author}{Tahbaz-Salehi, A.}, \bibinfo{year}{2015}a.
\newblock \bibinfo{title}{Networks, shocks, and systemic risk}.
\newblock \bibinfo{type}{Technical Report}. National Bureau of Economic
  Research.
\bibitem[{Acemoglu et~al.(2015b)Acemoglu, Ozdaglar and
  Tahbaz-Salehi}]{acemoglu2015systemic}
\bibinfo{author}{Acemoglu, D.}, \bibinfo{author}{Ozdaglar, A.},
  \bibinfo{author}{Tahbaz-Salehi, A.}, \bibinfo{year}{2015}b.
\newblock \bibinfo{title}{Systemic risk and stability in financial networks}.
\newblock \bibinfo{journal}{American Economic Review} \bibinfo{volume}{105},
  \bibinfo{pages}{564--608}.
\bibitem[{Aliprantis and Border(2006)}]{aliprantis2006infinite}
\bibinfo{author}{Aliprantis, C.D.}, \bibinfo{author}{Border, K.C.},
  \bibinfo{year}{2006}.
\newblock \bibinfo{title}{Infinite Dimensional Analysis: A Hitchhiker's Guide}.
\newblock \bibinfo{publisher}{Springer}.
\bibitem[{Balbus(2016)}]{balbus2016non}
\bibinfo{author}{Balbus, L.}, \bibinfo{year}{2016}.
\newblock \bibinfo{title}{On non-negative recursive utilities in dynamic
  programming with nonlinear aggregator and ces}.
\newblock \bibinfo{journal}{University of Zielora G{\'o}ra Working Paper} .
\bibitem[{Balbus et~al.(2013)Balbus, Reffett and
  Wo{\'z}ny}]{balbus2013constructive}
\bibinfo{author}{Balbus, {\L}.}, \bibinfo{author}{Reffett, K.},
  \bibinfo{author}{Wo{\'z}ny, {\L}.}, \bibinfo{year}{2013}.
\newblock \bibinfo{title}{A constructive geometrical approach to the uniqueness
  of markov stationary equilibrium in stochastic games of intergenerational
  altruism}.
\newblock \bibinfo{journal}{Journal of Economic Dynamics and Control}
  \bibinfo{volume}{37}, \bibinfo{pages}{1019--1039}.
\bibitem[{Baldwin and Venables(2013)}]{baldwin2013spiders}
\bibinfo{author}{Baldwin, R.}, \bibinfo{author}{Venables, A.J.},
  \bibinfo{year}{2013}.
\newblock \bibinfo{title}{Spiders and snakes: offshoring and agglomeration in
  the global economy}.
\newblock \bibinfo{journal}{Journal of International Economics}
  \bibinfo{volume}{90}, \bibinfo{pages}{245--254}.
\bibitem[{Becker and Murphy(1992)}]{becker1992division}
\bibinfo{author}{Becker, G.S.}, \bibinfo{author}{Murphy, K.M.},
  \bibinfo{year}{1992}.
\newblock \bibinfo{title}{The division of labor, coordination costs, and
  knowledge}.
\newblock \bibinfo{journal}{The Quarterly Journal of Economics}
  \bibinfo{volume}{107}, \bibinfo{pages}{1137--1160}.
\bibitem[{Becker and Rinc{\'o}n-Zapatero(2017)}]{becker2017recursive}
\bibinfo{author}{Becker, R.A.}, \bibinfo{author}{Rinc{\'o}n-Zapatero, J.P.},
  \bibinfo{year}{2017}.
\newblock \bibinfo{title}{Recursive utiity and thompson aggregators} .
\bibitem[{Bellman(1957)}]{bellman1957dynamic}
\bibinfo{author}{Bellman, R.E.}, \bibinfo{year}{1957}.
\newblock \bibinfo{title}{Dynamic Programming}.
\newblock \bibinfo{publisher}{Princeton University Press}.
\bibitem[{Bessaga(1959)}]{bessaga1959converse}
\bibinfo{author}{Bessaga, C.}, \bibinfo{year}{1959}.
\newblock \bibinfo{title}{On the converse of banach “fixed-point
  principle”}.
\newblock \bibinfo{journal}{Colloquium Mathematicum} \bibinfo{volume}{7},
  \bibinfo{pages}{41–43}.
\newblock \URLprefix \url{http://dx.doi.org/10.4064/cm-7-1-41-43},
  \DOIprefix\doi{10.4064/cm-7-1-41-43}.
\bibitem[{Bigio and La’O(2016)}]{bigio2016financial}
\bibinfo{author}{Bigio, S.}, \bibinfo{author}{La’O, J.},
  \bibinfo{year}{2016}.
\newblock \bibinfo{title}{Financial frictions in production networks}.
\newblock \bibinfo{type}{Technical Report}. National Bureau of Economic
  Research.
\bibitem[{Bloise and Vailakis(2018)}]{bloise2018convex}
\bibinfo{author}{Bloise, G.}, \bibinfo{author}{Vailakis, Y.},
  \bibinfo{year}{2018}.
\newblock \bibinfo{title}{Convex dynamic programming with (bounded) recursive
  utility}.
\newblock \bibinfo{journal}{Journal of Economic Theory} \bibinfo{volume}{173},
  \bibinfo{pages}{118--141}.
\bibitem[{Borovi{\v{c}}ka and Stachurski(2017)}]{borovivcka2017necessary}
\bibinfo{author}{Borovi{\v{c}}ka, J.}, \bibinfo{author}{Stachurski, J.},
  \bibinfo{year}{2017}.
\newblock \bibinfo{title}{Necessary and sufficient conditions for existence and
  uniqueness of recursive utilities}.
\newblock \bibinfo{type}{Technical Report}. National Bureau of Economic
  Research.
\bibitem[{Borovi{\v{c}}ka and Stachurski(2018)}]{borovivcka2018existence}
\bibinfo{author}{Borovi{\v{c}}ka, J.}, \bibinfo{author}{Stachurski, J.},
  \bibinfo{year}{2018}.
\newblock \bibinfo{title}{Existence and uniqueness of equilibrium asset prices
  over infinite horizons}.
\newblock \bibinfo{type}{Technical Report}.
\bibitem[{Carvalho(2007)}]{carvalho2007aggregate}
\bibinfo{author}{Carvalho, V.}, \bibinfo{year}{2007}.
\newblock \bibinfo{title}{Aggregate fluctuations and the network structure of
  intersectoral trade} .
\bibitem[{Cheney(2013)}]{cheney2013analysis}
\bibinfo{author}{Cheney, W.}, \bibinfo{year}{2013}.
\newblock \bibinfo{title}{Analysis for applied mathematics}. volume
  \bibinfo{volume}{208}.
\newblock \bibinfo{publisher}{Springer Science \& Business Media}.
\bibitem[{Choi(1994)}]{choi1994medians}
\bibinfo{author}{Choi, K.P.}, \bibinfo{year}{1994}.
\newblock \bibinfo{title}{On the medians of gamma distributions and an equation
  of ramanujan}.
\newblock \bibinfo{journal}{Proceedings of the American Mathematical Society}
  \bibinfo{volume}{121}, \bibinfo{pages}{245--251}.
\bibitem[{Ciccone(2002)}]{ciccone2002input}
\bibinfo{author}{Ciccone, A.}, \bibinfo{year}{2002}.
\newblock \bibinfo{title}{Input chains and industrialization}.
\newblock \bibinfo{journal}{The Review of Economic Studies}
  \bibinfo{volume}{69}, \bibinfo{pages}{565--587}.
\bibitem[{Coase(1937)}]{coase1937nature}
\bibinfo{author}{Coase, R.H.}, \bibinfo{year}{1937}.
\newblock \bibinfo{title}{The nature of the firm}.
\newblock \bibinfo{journal}{Economica} \bibinfo{volume}{4},
  \bibinfo{pages}{386--405}.
\bibitem[{Coleman(1991)}]{coleman1991equilibrium}
\bibinfo{author}{Coleman, W.J.}, \bibinfo{year}{1991}.
\newblock \bibinfo{title}{Equilibrium in a production economy with an income
  tax}.
\newblock \bibinfo{journal}{Econometrica: Journal of the Econometric Society} ,
  \bibinfo{pages}{1091--1104}.
\bibitem[{Coleman(2000)}]{coleman2000uniqueness}
\bibinfo{author}{Coleman, W.J.}, \bibinfo{year}{2000}.
\newblock \bibinfo{title}{Uniqueness of an equilibrium in infinite-horizon
  economies subject to taxes and externalities}.
\newblock \bibinfo{journal}{Journal of Economic Theory} \bibinfo{volume}{95},
  \bibinfo{pages}{71--78}.
\bibitem[{Datta et~al.(2002a)Datta, Mirman, Morand and
  Reffett}]{datta2002monotone}
\bibinfo{author}{Datta, M.}, \bibinfo{author}{Mirman, L.J.},
  \bibinfo{author}{Morand, O.F.}, \bibinfo{author}{Reffett, K.L.},
  \bibinfo{year}{2002}a.
\newblock \bibinfo{title}{Monotone methods for markovian equilibrium in dynamic
  economies}.
\newblock \bibinfo{journal}{Annals of Operations Research}
  \bibinfo{volume}{114}, \bibinfo{pages}{117--144}.
\bibitem[{Datta et~al.(2002b)Datta, Mirman and Reffett}]{datta2002existence}
\bibinfo{author}{Datta, M.}, \bibinfo{author}{Mirman, L.J.},
  \bibinfo{author}{Reffett, K.L.}, \bibinfo{year}{2002}b.
\newblock \bibinfo{title}{Existence and uniqueness of equilibrium in distorted
  dynamic economies with capital and labor}.
\newblock \bibinfo{journal}{Journal of Economic Theory} \bibinfo{volume}{103},
  \bibinfo{pages}{377--410}.
\bibitem[{Dedrick et~al.(2011)Dedrick, Kraemer and
  Linden}]{dedrick2011distribution}
\bibinfo{author}{Dedrick, J.}, \bibinfo{author}{Kraemer, K.L.},
  \bibinfo{author}{Linden, G.}, \bibinfo{year}{2011}.
\newblock \bibinfo{title}{The distribution of value in the mobile phone supply
  chain}.
\newblock \bibinfo{journal}{Telecommunications Policy} \bibinfo{volume}{35},
  \bibinfo{pages}{505--521}.
\bibitem[{Du(1989)}]{du1989fixed}
\bibinfo{author}{Du, Y.}, \bibinfo{year}{1989}.
\newblock \bibinfo{title}{Fixed points of a class of non-compact operators and
  applications}.
\newblock \bibinfo{journal}{Acta Mathematica Sinica} \bibinfo{volume}{32},
  \bibinfo{pages}{618--627}.
\bibitem[{Guo et~al.(2004)Guo, Cho and Zhu}]{guo2004partial}
\bibinfo{author}{Guo, D.}, \bibinfo{author}{Cho, Y.J.}, \bibinfo{author}{Zhu,
  J.}, \bibinfo{year}{2004}.
\newblock \bibinfo{title}{Partial ordering methods in nonlinear problems}.
\newblock \bibinfo{publisher}{Nova Publishers}.
\bibitem[{Guo and Lakshmikantham(1988)}]{guo1988nonlinear}
\bibinfo{author}{Guo, D.}, \bibinfo{author}{Lakshmikantham, V.},
  \bibinfo{year}{1988}.
\newblock \bibinfo{title}{Nonlinear problems in abstract cones}.
\newblock \bibinfo{publisher}{Academic Press}.
\newblock \DOIprefix\doi{https://doi.org/10.1016/C2013-0-10750-7}.
\bibitem[{Janos(1967)}]{janos1967converse}
\bibinfo{author}{Janos, L.}, \bibinfo{year}{1967}.
\newblock \bibinfo{title}{A converse of banach's contraction theorem}.
\newblock \bibinfo{journal}{Proceedings of the American Mathematical Society}
  \bibinfo{volume}{18}, \bibinfo{pages}{287--289}.
\bibitem[{Jones(2011)}]{jones2011intermediate}
\bibinfo{author}{Jones, C.I.}, \bibinfo{year}{2011}.
\newblock \bibinfo{title}{Intermediate goods and weak links in the theory of
  economic development}.
\newblock \bibinfo{journal}{American Economic Journal: Macroeconomics}
  \bibinfo{volume}{3}, \bibinfo{pages}{1--28}.
\bibitem[{Kikuchi et~al.(2018)Kikuchi, Nishimura and
  Stachurski}]{kikuchi2018span}
\bibinfo{author}{Kikuchi, T.}, \bibinfo{author}{Nishimura, K.},
  \bibinfo{author}{Stachurski, J.}, \bibinfo{year}{2018}.
\newblock \bibinfo{title}{Span of control, transaction costs, and the structure
  of production chains}.
\newblock \bibinfo{journal}{Theoretical Economics} \bibinfo{volume}{13},
  \bibinfo{pages}{729--760}.
\bibitem[{Kraemer et~al.(2011)Kraemer, Linden and
  Dedrick}]{kraemer2011capturing}
\bibinfo{author}{Kraemer, K.L.}, \bibinfo{author}{Linden, G.},
  \bibinfo{author}{Dedrick, J.}, \bibinfo{year}{2011}.
\newblock \bibinfo{title}{Capturing value in global networks: Apple’s ipad
  and iphone}.
\newblock \bibinfo{journal}{University of California, Irvine, University of
  California, Berkeley, y Syracuse University, NY. http://pcic. merage. uci.
  edu/papers/2011/value\_iPad\_iPhone. pdf. Consultado el}
  \bibinfo{volume}{15}.
\bibitem[{Krasnosel'skii(1964)}]{krasnoselskii1964positive}
\bibinfo{author}{Krasnosel'skii}, \bibinfo{year}{1964}.
\newblock \bibinfo{title}{Positive Solutions of Operator Equations}.
\bibitem[{Krasnosel'skii and Zabre\u{i}ko(1984)}]{krasnosel1984geometrical}
\bibinfo{author}{Krasnosel'skii, M.}, \bibinfo{author}{Zabre\u{i}ko, P.},
  \bibinfo{year}{1984}.
\newblock \bibinfo{title}{Geometrical methods of nonlinear analysis}.
\newblock Grundlehren der mathematischen Wissenschaften,
  \bibinfo{publisher}{Springer-Verlag}.
\newblock \URLprefix \url{https://books.google.com.au/books?id=8Q2oAAAAIAAJ}.
\bibitem[{Krasnosel’skii et~al.(1972)Krasnosel’skii, Vainikko, Zabreiko,
  Rutitskii and Stetsenko}]{krasnoselskii1972approximate}
\bibinfo{author}{Krasnosel’skii, M.A.}, \bibinfo{author}{Vainikko, G.M.},
  \bibinfo{author}{Zabreiko, P.P.}, \bibinfo{author}{Rutitskii, Y.B.},
  \bibinfo{author}{Stetsenko, V.Y.}, \bibinfo{year}{1972}.
\newblock \bibinfo{title}{Approximate Solution of Operator Equations}.
\newblock \bibinfo{publisher}{Springer Netherlands}.
\newblock \URLprefix \url{http://dx.doi.org/10.1007/978-94-010-2715-1},
  \DOIprefix\doi{10.1007/978-94-010-2715-1}.
\bibitem[{Lacker and Schreft(1991)}]{lacker1991money}
\bibinfo{author}{Lacker, J.M.}, \bibinfo{author}{Schreft, S.},
  \bibinfo{year}{1991}.
\newblock \bibinfo{title}{Money, trade credit and asset prices} .
\bibitem[{Leader(1982)}]{leader1982uniformly}
\bibinfo{author}{Leader, S.}, \bibinfo{year}{1982}.
\newblock \bibinfo{title}{Uniformly contractive fixed points in compact metric
  spaces}.
\newblock \bibinfo{journal}{Proceedings of the American Mathematical Society}
  \bibinfo{volume}{86}, \bibinfo{pages}{153--158}.
\bibitem[{Levine(2012)}]{levine2012production}
\bibinfo{author}{Levine, D.K.}, \bibinfo{year}{2012}.
\newblock \bibinfo{title}{Production chains}.
\newblock \bibinfo{journal}{Review of Economic Dynamics} \bibinfo{volume}{15},
  \bibinfo{pages}{271--282}.
\bibitem[{Lucas(1978)}]{lucas1978size}
\bibinfo{author}{Lucas, R.E.}, \bibinfo{year}{1978}.
\newblock \bibinfo{title}{On the size distribution of business firms}.
\newblock \bibinfo{journal}{The Bell Journal of Economics} ,
  \bibinfo{pages}{508--523}.
\bibitem[{Marinacci and Montrucchio(2010)}]{marinacci2010unique}
\bibinfo{author}{Marinacci, M.}, \bibinfo{author}{Montrucchio, L.},
  \bibinfo{year}{2010}.
\newblock \bibinfo{title}{Unique solutions for stochastic recursive utilities}.
\newblock \bibinfo{journal}{Journal of Economic Theory} \bibinfo{volume}{145},
  \bibinfo{pages}{1776--1804}.
\bibitem[{Marinacci and Montrucchio(2017)}]{marinacci2017unique}
\bibinfo{author}{Marinacci, M.}, \bibinfo{author}{Montrucchio, L.},
  \bibinfo{year}{2017}.
\newblock \bibinfo{title}{Unique Tarski fixed points}.
\newblock \bibinfo{type}{Technical Report}.
\bibitem[{Morand and Reffett(2003)}]{morand2003existence}
\bibinfo{author}{Morand, O.F.}, \bibinfo{author}{Reffett, K.L.},
  \bibinfo{year}{2003}.
\newblock \bibinfo{title}{Existence and uniqueness of equilibrium in nonoptimal
  unbounded infinite horizon economies}.
\newblock \bibinfo{journal}{Journal of Monetary Economics}
  \bibinfo{volume}{50}, \bibinfo{pages}{1351--1373}.
\bibitem[{Pavoni et~al.(2018)Pavoni, Sleet and Messner}]{pavoni2018dual}
\bibinfo{author}{Pavoni, N.}, \bibinfo{author}{Sleet, C.},
  \bibinfo{author}{Messner, M.}, \bibinfo{year}{2018}.
\newblock \bibinfo{title}{The dual approach to recursive optimization: theory
  and examples}.
\newblock \bibinfo{journal}{Econometrica} \bibinfo{volume}{86},
  \bibinfo{pages}{133--172}.
\bibitem[{Ren and Stachurski(2018)}]{ren2018dynamic}
\bibinfo{author}{Ren, G.}, \bibinfo{author}{Stachurski, J.},
  \bibinfo{year}{2018}.
\newblock \bibinfo{title}{Dynamic programming with recursive preferences:
  Optimality and applications}.
\newblock \bibinfo{journal}{arXiv preprint arXiv:1812.05748} .
\bibitem[{Rinc{\'o}n-Zapatero and
  Rodr{\'i}guez-Palmero(2003)}]{rincon2003existence}
\bibinfo{author}{Rinc{\'o}n-Zapatero, J.P.},
  \bibinfo{author}{Rodr{\'i}guez-Palmero, C.}, \bibinfo{year}{2003}.
\newblock \bibinfo{title}{Existence and uniqueness of solutions to the bellman
  equation in the unbounded case}.
\newblock \bibinfo{journal}{Econometrica} \bibinfo{volume}{71},
  \bibinfo{pages}{1519--1555}.
\bibitem[{Thompson(1963)}]{thompson1963certain}
\bibinfo{author}{Thompson, A.C.}, \bibinfo{year}{1963}.
\newblock \bibinfo{title}{On certain contraction mappings in a partially
  ordered vector space}.
\newblock \bibinfo{journal}{Proceedings of the American Mathematical Society}
  \bibinfo{volume}{14}, \bibinfo{pages}{438--443}.
\bibitem[{Williamson and Janos(1987)}]{williamson1987constructing}
\bibinfo{author}{Williamson, R.}, \bibinfo{author}{Janos, L.},
  \bibinfo{year}{1987}.
\newblock \bibinfo{title}{Constructing metrics with the heine-borel property}.
\newblock \bibinfo{journal}{Proceedings of the American Mathematical Society}
  \bibinfo{volume}{100}, \bibinfo{pages}{567--573}.
\bibitem[{Zhang(2013)}]{zhang2012variational}
\bibinfo{author}{Zhang, Z.}, \bibinfo{year}{2013}.
\newblock \bibinfo{title}{Variational, Topological, and Partial Order Methods
  with Their Applications}. volume~\bibinfo{volume}{29} of
  \textit{\bibinfo{series}{Developments in Mathematics}}.
\newblock \bibinfo{publisher}{Springer Berlin Heidelberg}.
\newblock \DOIprefix\doi{10.1007/978-3-642-30709-6}.

\end{thebibliography}
\end{document}